\documentclass[a4paper,UKenglish,cleveref, autoref, thm-restate,numberwithinsect]{lipics-v2021}

\pdfoutput=1 %
\hideLIPIcs  %

\graphicspath{{./graphics/}}%

\title{Paged Geophylogenies: A Coloring Approach to External Labeling with Tree Constraints} %

\titlerunning{Paged Geophylogenies} %

\author{Thomas Depian}{TU Wien, Austria}{tdepian@ac.tuwien.ac.at}{https://orcid.org/0009-0003-7498-6271}{Project No.\ 10.47379/ICT22029 of the Vienna Science and Technology Fund (WWTF).}

\author{Thomas C. van Dijk}{TU Eindhoven, The Netherlands}{t.c.v.dijk@tue.nl}{https://orcid.org/0000-0001-6553-7317}{}

\author{Martin N{\"o}llenburg}{TU Wien, Austria}{noellenburg@ac.tuwien.ac.at}{https://orcid.org/0000-0003-0454-3937}{Project No.\ 10.47379/ICT22029 of the Vienna Science and Technology Fund (WWTF).}

\authorrunning{T. Depian, T.\,C. van Dijk and M. N{\"o}llenburg} %

\Copyright{Thomas Depian, Thomas C. van Dijk and Martin N{\"o}llenburg} %



\supplementdetails{Source Code}{https://doi.org/10.17605/OSF.IO/RBQD3}

\nolinenumbers %

\EventEditors{Maarten L\"{o}ffler and Silvia Miksch}
\EventNoEds{2}
\EventLongTitle{34th International Symposium on Graph Drawing and Network Visualization (GD 2026)}
\EventShortTitle{GD 2026}
\EventAcronym{GD}
\EventYear{2026}
\EventDate{August 19--21, 2026}
\EventLocation{Ontario, Canada}
\EventLogo{}
\SeriesVolume{396}
\ArticleNo{8}
\usepackage{mathtools}
\usepackage{complexity}
\usepackage{todonotes}
\usepackage{xspace}
\usepackage{pifont}
\usepackage{cite}

\usepackage{xcolor}
\definecolor{defblue}{rgb}{0.121,0.47,0.705}
\hypersetup{colorlinks=true,
    linkcolor=defblue,
    anchorcolor=defblue,
    citecolor=defblue,
    filecolor=defblue,
    menucolor=defblue,
    urlcolor=defblue,
    bookmarksopen=true,
    bookmarksopenlevel=2,
    bookmarksnumbered=true,
    plainpages=false
    }

\crefname{Con}{Con}{Cons}
\Crefname{Con}{Con}{Cons}

\theoremstyle{claimstyle}
\newtheorem{obclaim}[claim]{Observation}

\newboolean{long}
\newcommand{\AppendixSymbol}{\ding{72}}
\setboolean{long}{true} %
\ifthenelse{\boolean{long}}{
	\NewDocumentEnvironment{prooflater}{m}{\begin{proof}}{\end{proof}\ignorespacesafterend}
	\NewDocumentEnvironment{proofsketch}{o +b}{}{\ignorespacesafterend}
	\newcommand{\restateref}[1]{}
	\NewDocumentEnvironment{statelater}{m}{}{}

	\NewDocumentCommand{\onlyShort}{+m}{}
	\NewDocumentCommand{\onlyLong}{+m}{#1}
	\NewDocumentCommand{\shortLong}{+m +m}{#2}
}{
	\NewDocumentEnvironment{prooflater}{m +b}{%
		\expandafter\global\expandafter\def\csname#1\endcsname{\begin{proof}#2\end{proof}}%
	}{\ignorespacesafterend}
	\NewDocumentEnvironment{proofsketch}{O{Proof sketch.}}{\begin{proof}[#1]}{\end{proof}\ignorespacesafterend}
	\usepackage{apptools}

	\newcommand{\restateref}[1]{[\IfAppendix{\AppendixSymbol{}}{\AppendixSymbol{}}]}
	
	\NewDocumentEnvironment{statelater}{m +b}{%
		\expandafter\global\expandafter\def\csname#1\endcsname{#2}%
	}{\ignorespacesafterend}

	\NewDocumentCommand{\onlyShort}{+m}{#1}
	\NewDocumentCommand{\onlyLong}{+m}{}
	\NewDocumentCommand{\shortLong}{+m +m}{#1}
}
\let\oldrestatable\restatable
\def\restatable{\expandafter\oldrestatable}

\NewDocumentCommand{\parent}{o m}{\ensuremath{\text{pa}\IfNoValueF{#1}{_{#1}}(#2)}\xspace}
\NewDocumentCommand{\children}{o m}{\ensuremath{\text{ch}\IfNoValueF{#1}{_{#1}}(#2)}\xspace}
\NewDocumentCommand{\treeroot}{o}{\ensuremath{\rho\IfNoValueF{#1}{(#1)}}\xspace}
\NewDocumentCommand{\lca}{o m m}{\ensuremath{\text{lca}\IfNoValueF{#1}{_{#1}}(#2, #3)}\xspace}

\newcommand{\labeling}{\ensuremath{\mathcal{L}}\xspace}
\newcommand{\leader}{\ensuremath{\lambda}\xspace}

\newcommand{\instance}{\ensuremath{G}\xspace}
\newcommand{\instancelong}{\ensuremath{(T, S, R)}\xspace}

\newcommand{\ILPCrossing}{\textsf{ILPCross}\xspace}
\newcommand{\ILPColors}{\textsf{ILPColors}\xspace}
\newcommand{\ILPColorsSym}{\textsf{ILPColorsSym}\xspace}
\newcommand{\ILPColorsRep}{\textsf{ILPColorsRep}\xspace}
\newcommand{\ILPColorsPop}{\textsf{ILPColorsPop}\xspace}

\newcommand{\Size}[1]{\ensuremath{\left\vert #1 \right\vert}}
\newcommand{\BigO}[1]{\ensuremath{\mathcal{O}(#1)}}

\newcommand{\NewText}[1]{{\color{black}#1}}

\begin{document}

\maketitle

\begin{abstract}
\emph{Geophylogenies} are a common type of diagram for visualizing the evolutionary history of species in a geographic context.
As a drawing problem, these diagrams are commonly modeled as a rooted binary ``phylogenetic'' tree $T$ where every leaf is associated with a point feature (``site'') in a rectangular map range.
The tree is drawn downward planar, its leaves are placed at equidistant positions on the upper boundary of the map, and each leaf is connected to its corresponding site by a straight-line leader.
Prior work focuses on minimizing the number of leader crossings, or avoiding leaders altogether.
In this paper, we explore \emph{paged} geophylogenies, where the leaders can be partitioned into multiple pages where only crossings within a page are counted.
For the general case, where each page can contain an arbitrary subset of the leaders, we provide an integer linear programming (ILP) formulation for minimizing the number of pages, and a polynomial-time algorithm for a special case that is equivalent to a one-sided tanglegram.
We argue that, from a visualization perspective, instead each page must contain only leaders for a single subtree, and provide an \BigO{n^7} time algorithm for minimizing pages in this setting -- which also involves improving the fastest known algorithm for testing the existence of a crossing-free labeling.
Counter to what our worst case bound suggests, our implementation can handle instances with hundreds of sites within a second, as we show in our experimental evaluation, which further investigates the trade-offs between leader crossings and the number of pages.

\end{abstract}

\section{Introduction}
Labeling feature points on a map is a recurring and time-consuming task in the design of scientific illustrations and geographic information systems~\cite{NNR.RCL.2017,BNN.ELF.2021}.
Usually, \emph{labels} are placed close to their feature point (also called \emph{site})~\cite{AvKS.Lpm.1998,vKSW.Pls.1999},
but this becomes cluttered or infeasible if there are many sites and\NewText{, furthermore,} obscures background features~\cite{Bri.RGS.1990}.
\emph{External labeling} avoids this by placing the labels outside the map and connecting them to their sites using \emph{leader} lines~\cite{BNN.ELF.2021}.

Existing work often focuses on \emph{boundary labeling}, which is a special case of external labeling in which the labels are placed along the (rectangular) boundary of the map and leaders must be crossing-free~\cite{BHKN.AMC.2009,KLW.LBL.2014,GHN.MBL.2015}.
Efficient algorithms for finding a labeling with, e.g., minimum leader length are known~\cite{BKSW.BlM.2007,BHKN.AMC.2009}, and the literature proposes various approaches to place the labels~\cite{HPL.BLF.2014,BNT+.BLC.2024} or leader styles~\cite{BGNN.rlb.2019,BKNS.BLO.2010}; see also the book by Bekos, Niedermann, and Nöllenburg~\cite{BNN.ELF.2021}.
In recent years, there %
\NewText{has been} growing interest in the (algorithmic) study of boundary labeling under additional constraints.
In particular, Niedermann, Nöllenburg, and Rutter~\cite{NNR.RCL.2017} suggested incorporating grouping constraints into the model and Depian, Nöllenburg, Terziadis, and Wallinger~\cite{DNTW.Cbl.2025,DNTW.CBL.2024} systematically studied boundary labeling under grouping and ordering constraints.
The existence of a crossing-free labeling in the presence of such constraints is no longer guaranteed and finding one is \NP-\NewText{complete} in general~\cite{DNTW.Cbl.2025}.

Klawitter, Klesen, Scholl, Van Dijk, and Zaft~\cite{KKS+.VGI.2023,KKS+.VGI.2025} study a specific kind of hierarchical grouping constraints arising from drawing phylogenetic trees.
A \emph{phylogenetic tree} is a rooted binary tree $T$ that models the evolutionary history of \emph{taxa}\NewText{, i.e., groups of organisms}, which are represented by the leaves $L(T)$ of $T$.
Drawings of phylogenetic trees are, with some distinctions, also known as \emph{cladograms} or \emph{dendrograms}.
Klawitter et al.\ study the problem of drawing a \emph{geophylogeny}~$G$ where each leaf has an associated geographic location.
They formalize it as follows: $\instance = \instancelong$ consists of a phylogenetic tree $T$ where each leaf is associated with a distinct site from $S$ that lies inside the map region $R$ which could, for example, represent the location where a taxon was found; see \Cref{fig:example} for an example from the biology literature.
They consider downward planar drawings of $T$, i.e., \emph{embeddings} of $T$, such that the $n$ leaves $L(T)$ are uniformly spaced along the top boundary of $R$.
Different embeddings of the tree induce different labelings, and they study the problem of finding a good embedding for various optimization goals using straight-line and L-shaped leaders.
In particular, they present an $\BigO{n^6}$-time algorithm to determine whether a crossing-free labeling exists and show that crossing-minimization is in general \NP-hard~\cite{KKS+.VGI.2025}.
Furthermore, they provide parameterized and greedy algorithms along with an ILP formulation for the crossing-minimization problem.%
\begin{figure}
    \centering
    \includegraphics[width=\linewidth]{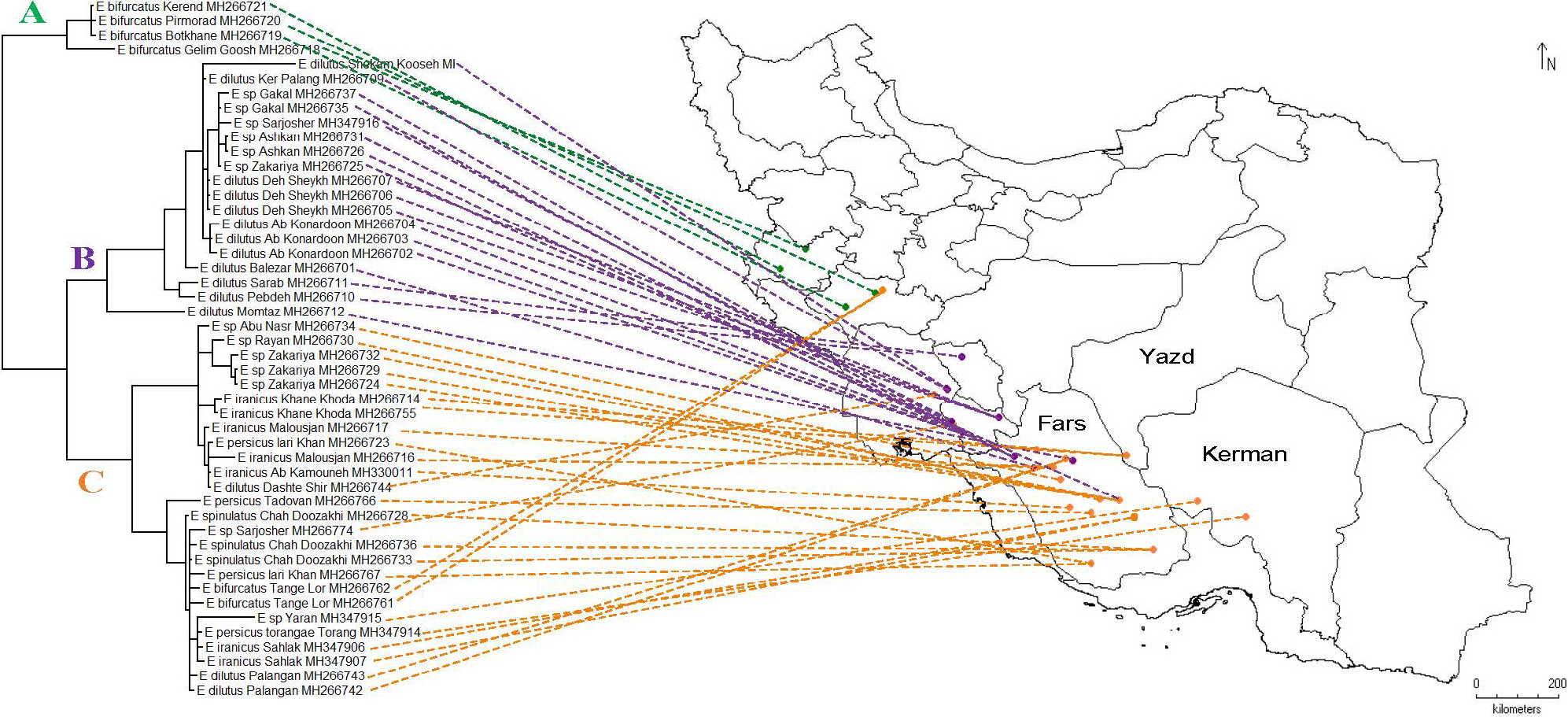}
    \caption{Drawing of a geophylogeny depicting \emph{Eremogryllodes}.
    Notice the use of color for three subtrees.
    Figure reproduced from Tahami, Hojat-Ansari, Mina, Namyatova, and Sadeghi~\cite{Tahami2021}.}
    \label{fig:example}
\end{figure}

While crossing minimization is a standard and well-justified goal to improve the readability of a drawing~\cite{Pur.WAh.1997} (or, as in our case, of a labeling), it is far from the only approach.
Consider, for example, \Cref{fig:intro}b, which shows a crossing-optimal labeling of a real-world geophylogeny.
Although it is an improvement over the unoptimized labeling from \Cref{fig:intro}a as it occurred in the biology literature, following individual leaders in the center of the map can be difficult.
\emph{Beyond planarity}~\cite{DLM.SGD.2019} is a relatively young concept in graph drawing that attempts to deal with drawings where crossings are unavoidable.
In particular, the concept of \emph{thickness} has been studied extensively, where edges are colored with the smallest possible number of colors such that there are no monochromatic crossings~\cite{DM.RGT.2018,MOS.TGS.1998,JRRS.GT2.2023,dffgn-ptgt-25}. %
However, to the best of our knowledge, applying this methodology to boundary labeling problems is novel.

With this paper, we build a bridge between these two fields of graph drawing.
More concretely, we apply the thickness metaphor to the problem of labeling geophylogenies:
for a given geophylogeny $\instance = \instancelong$, we want to find an embedding of the phylogenetic tree~$T$ together with a coloring $\varphi$ of its leaves such that the induced labeling using straight-line leaders has no monochromatic crossings, where the color of a leader is inherited from the leaf it attaches to; see \Cref{fig:intro}c for an example.
Note that an optimal labeling could have more crossings than a crossing-minimal labeling -- but no two leaders of the same color cross.

\NewText{Our approach is} closely related to Gedicke, Bonerath, Niedermann, and Haunert~\cite{GBNH.ZME.2021}, who studied the problem of labeling feature-dense maps by distributing the labels across several independent \emph{pages}, formed by the color classes in our model.
\Cref{fig:intro}c has one particular downside in light of the biological application.
The color classes do not correspond with the (biologically meaningful) hierarchy of the phylogenetic tree $T$: two closely related taxa can be on different pages whereas two unrelated taxa are on the same page.
It would be more appropriate if every page features precisely a subtree of $T$, i.e., if the coloring of the leaves is \emph{consistent} with $T$ as in Figures~\ref{fig:example} and \ref{fig:intro}d.
\NewText{
In particular for \emph{tree-consistent} colorings, interpreting colors as individual pages enables two different views on the same labeling:
On the one hand, we can draw the entire labeling at once, i.e., overlay differently-colored leaders.
Here, colors can increase the readability of such \emph{overview drawings} as each color class is crossing-free; compare also Figures~\ref{fig:intro}b and \ref{fig:intro}d.
On the other hand, we can partition the labeling along the colors into literal pages.
Here, each page, i.e., color, focuses on a different subtree of $T$; compare also Figures~\ref{fig:intro}d and \ref{fig:intro}e.
From the perspective of the visualization, these two approaches are substantially different:
For the former approach, research on color use in visualizations puts an upper limit of around seven~\cite{Mac.Uce.1999} to at most ten~\cite{War.CFC.2021} on the number of colors that should be used when reliably encoding nominal data.
The latter approach can potentially handle a larger number of ``colors'' but for readability it still makes sense to minimize them. %
From an algorithmic perspective, there is no difference between these approaches.
Therefore, we will use the terms \emph{color} and \emph{page} interchangeably to address a (crossing-free) subset of the labeling.

}

\begin{figure}
	\centering
	\includegraphics[page=1]{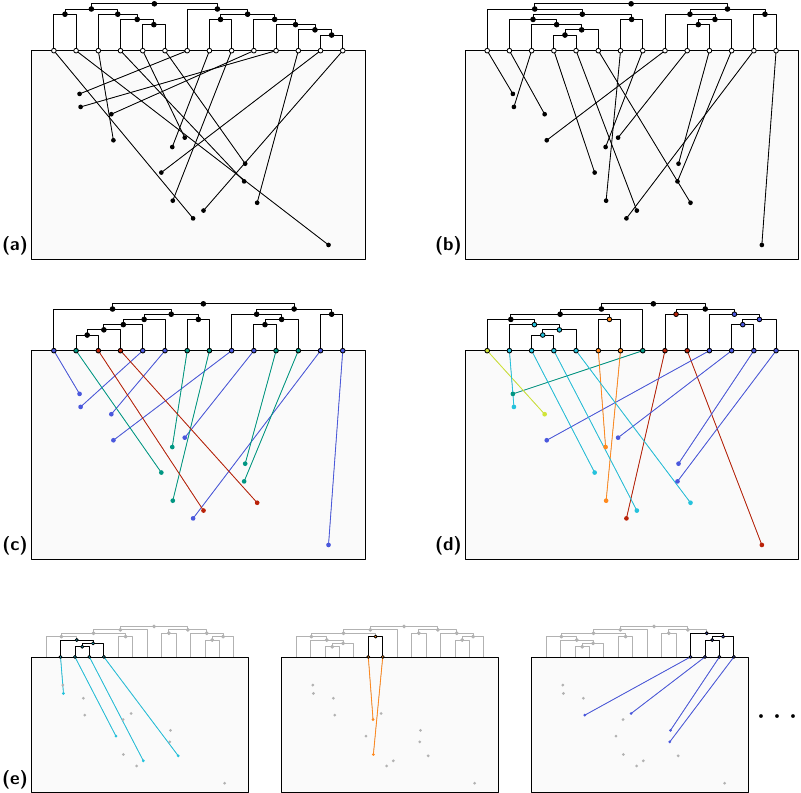}
	\caption{External labeling of a real-world geophylogeny on fourteen different fish species by Williams and Johnson (label text omitted for clarity). \textbf{\textsf{(a)}} Labeling induced by the embedding from~\cite{fish} \textbf{\textsf{(b)}} crossing-minimal labeling produced by the algorithm from Klawitter et al.~\cite{KKS+.VGI.2025}. 
		An optimal \textbf{\textsf{(c)}} unrestricted and \textbf{\textsf{(d)}}~tree-consistent coloring of the geophylogeny produced by the algorithms presented in this paper.
        \NewText{\textbf{\textsf{(e)}} Three pages of a tree-consistent coloring of the geophylogeny with a labeling for each color that contains only the sites represented by the respective subtree.
        Note that we maintain the same embedding across the different labelings, which we have taken from (d).}}
	\label{fig:intro}
\end{figure}

\subparagraph{Contributions.}
In this paper, we explore \emph{paged} geophylogenies, which formalize the above idea. %
We provide two efficient algorithms for practically relevant settings.
In \Cref{sec:tanglegram}, we consider so-called \emph{geometry-free} geophylogenies, where leader crossings are purely determined by the embedding of $T$; the exact geometric position of the sites are irrelevant~\cite{KKS+.VGI.2025}.
They can occur, for example, when the sites follow a river and are therefore almost 1-dimensionally spaced.
Since our polynomial-time algorithm equivalently finds a minimum-thickness embedding of one-sided tanglegrams~\cite{FKP.Ctv.2010}, we believe it may be of independent interest.
In \Cref{sec:tree-consistent}, we investigate tree-consistent colorings.
We present a dynamic programming algorithm that can find a %
minimum tree-consistent colored embedding of general geophylogenies in polynomial time. %
An important subroutine of our algorithm is a planarity test for sub-geophylogenies.
Therefore, we also revisit this special variant of our problem and improve the existing $\BigO{n^6}$-time algorithm to $\BigO{n^5}$.
The improvement is based on a new and arguably simpler technique to handle crossings in partial labelings.
We address in \Cref{sec:extensions} our initial motivation for this approach: In a tree-consistent coloring, each color class can represent a page that provides a planar labeling of a sub-geophylogeny.
With less information per page, more space becomes available, and we give polynomial-time algorithms for optimizing the quality of the labeling on each page, such as its leader length, by making better use of the available space.
In \Cref{sec:ilp}, we provide integer linear programming formulations %
and report in \Cref{sec:experiments} on an extensive experimental evaluation of the performance of our algorithms and the ``price'' of tree-consistency in terms of required colors.

\onlyShort{
\smallskip
\noindent
\emph{Full details for statements marked by \AppendixSymbol{} can be found in the %
\NewText{full version~\cite{ARXIV}}.}
}

\section{Preliminaries}
\label{sec:preliminaries}
For an integer $p \geq 1$, let $[p] \coloneqq \{1, 2, \ldots, p\}$.
See \Cref{fig:notation} for an example of the following notation, which follows previous work on labeling geophylogenies~\cite{KKS+.VGI.2025}.

\begin{figure}[t]
	\centering
	\includegraphics[page=1]{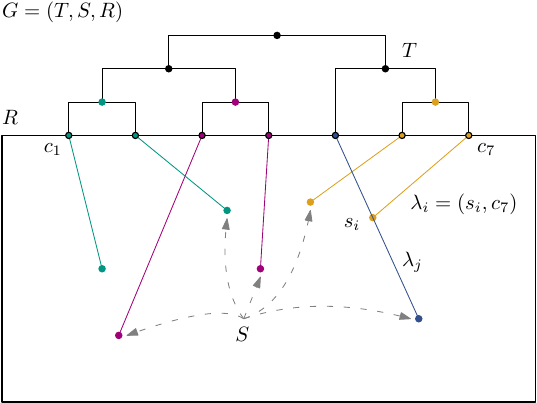}
	\caption{Example of a colored embedded geophylogeny $\instance = \instancelong$ explaining our notation.
    Observe that the leader $\leader_j$ crosses leader $\leader_i$ at its site $s_i$.}
	\label{fig:notation}
\end{figure}

\subparagraph*{Geophylogenies.}
Let $T$ be a phylogenetic tree on $n$ taxa, i.e., a binary tree with $n$ leaves $\ell_1, \ldots, \ell_n$.
We let $L(T)$ and $I(T)$ be the leaves and internal nodes of $T$ respectively, and $V(T) \coloneqq L(T) \cup I(T)$.
For an internal node $v \in I(T)$, its children are $\children[T]{v}$.
For any node $v \in V(T)$, let $T(v)$ denote the \emph{subtree} rooted at $v$ and define $n(v) \coloneqq \Size{L(T(v))}$.
In the following, we omit the reference to $T$ if it is clear from the context.
An \emph{embedding} gives an order of $\children[T]{v}$ for every $v \in I(T)$ and thereby defines a left-to-right order of $L(T)$; let $\Pi(T)$ be the set of leaf orders achievable by an embedding.
We let $\lca[T]{\ell_i}{\ell_j}$
denote the \emph{lowest common ancestor} of the two leaves $\ell_i$ and $\ell_j$ in $T$. %

A \emph{geophylogeny} $\instance = \instancelong$ consists of a phylogenetic tree $T$ together with a set $S$ of $n$ \emph{sites} $s_1, \ldots, s_n$ in $\mathbb{R}^2$ enclosed in an axis-aligned rectangle $R$ called the \emph{map}.
The sites have a one-to-one relation with $L(T)$ which we let, without loss of generality, be implicitly defined via the indices: leaders go from leaf $\ell_i$ to site $s_i$.
We will therefore refer to sites by their leaf and vice versa if there is no risk of confusion. 

\subparagraph*{Paged Labelings.}
Let $\instance$ be a geophylogeny and let $C(R) = \{c_1, \ldots, c_n\}$ be a set of~$n$ \emph{reference points} ordered left to right along the top boundary of $R$, spaced equidistantly, i.e., $x(c_1) < x(c_2) < \ldots < x(c_n)$; the $c_i$ are also called \emph{candidates} and we drop the reference to $R$ if it clear from the context.
A \emph{labeling} $\labeling \colon S \to C$
connects each $s_i \in S$ with the candidate $\labeling(s_i) = c_j$ using an \emph{s-leader} $\lambda_i = (s_i, c_j)$: a line segment from $s_i$ to~$c_j$.
We say that a site $s_i$ is labeled at the candidate $c_j$.
Note that a labeling $\labeling$ induces a left-to-right order $\pi$ of the sites based on the candidates they are labeled at.
The labeling is \emph{admissible} if $\pi \in \Pi(T)$.
Unless stated otherwise, we are interested in admissible labelings, which are uniquely defined (or: \emph{induced}) by an embedding $\pi \in \Pi(T)$ of $T$.
For an embedding $\pi$, let $\labeling_{\pi}$ denote the labeling induced by $\pi$.
A labeling is \emph{planar} if no two leaders cross, including at sites.

Let $\labeling$ be a labeling of $\instance$.
A $k$-\emph{coloring} $\varphi$ assigns to each %
leaf one of $k$ colors, such that there is no \emph{monochromatic crossing}, i.e., no crossing between two leaders $\lambda_i = (s_i, c_p)$ and $\lambda_j = (s_j, c_q)$ with $\varphi(\ell_i) = \varphi(\ell_j)$.
Since sites and leaves (at candidates) are connected via leaders, we equivalently refer to the color of sites and leaders.
A coloring is \emph{tree-consistent} if for every color $i \in [k]$ there exists a node $v \in V(T)$ such that precisely the leaves $L(T(v))$ have color $i$; %
see for example \Cref{fig:notation}, %
where internal nodes are colored to emphasize tree-consistency.
A \emph{colored embedding} $(\pi, \varphi)$ of \instance is an embedding $\pi \in \Pi(T)$ together with a coloring $\varphi$ of $L(T)$. %
It is \emph{minimum} if it uses the fewest colors among all colored embeddings.

\section{Unrestricted Colorings} %
\label{sec:tanglegram}

Since we are primarily interested in tree-consistent colorings, we only briefly discuss the unrestricted colored embedding problem for geophylogenies.
Given the \NP-hardness of coloring grounded segment graphs~\cite{DKMW.GLG.2023,Ung.kCC.1988}, %
we suspect our problem may be \NP-hard in general -- but leave this as an open problem.
Klawitter et al.\ define the \emph{s-region} of a site to be the triangle formed by itself, and the leftmost and rightmost candidate label positions:
a geophylogeny is \emph{geometry-free} if and only if no site is within another site's s-region.
In this section, we show that for geometry-free geophylogenies and, equivalently~\cite{KKS+.VGI.2025}, one-sided tanglegrams~\cite{FKP.Ctv.2010} the coloring problem can be solved in \BigO{n^3\log n} time; see \Cref{fig:tanglegram}.

\begin{figure}[t]
	\centering
	\includegraphics[page=1]{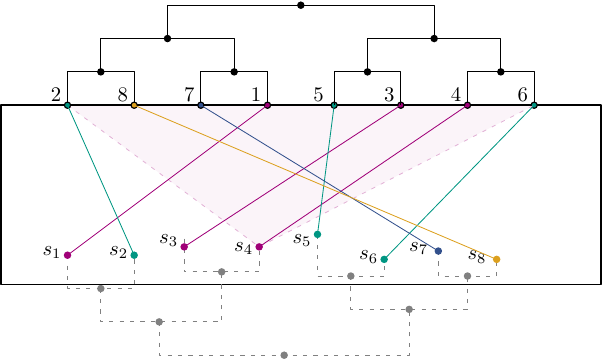}
	\caption{A colored embedded geometry-free geophylogeny. Observe that a longest decreasing subsequence $(8,7,5,4)$ defines the number of colors since their leaders pairwise cross. Also indicated are the $s$-region of site $s_4$ and the fixed tree in the corresponding tanglegram instance.}
	\label{fig:tanglegram}
\end{figure}

Without loss of generality, consider the leaves of $T$ to be bijectively labeled with the integers $[n]$ to indicate the left to right order of the sites (resp.\ the leaf order in the fixed tree in the tanglegram).
Recall that a labeling of $T$ is equivalent to a permutation $\pi\in\Pi(T)$ and an assignment of leaves to candidate label positions. %
In the geometry-free setting, leader crossings correspond to inversions in $\pi$.
Indeed, the intersection graph of the leaders is a permutation graph.
These are known to be perfect, and their largest clique is given by the longest decreasing subsequence in the permutation~\cite{Gol.AGT.1980}.
\begin{restatable}\restateref{thm:MLDS}{theorem}{thmMLDS}
    \label{thm:MLDS}
    Given a binary tree with $n$ leaves, which are are identified with the integers~$[n]$, the embedding that minimizes the longest left-to-right decreasing subsequence of leaf values can be found in \BigO{n^3\log n} time.
\end{restatable}

\newcommand{\snapup}[2]{\ensuremath{\mathrm{up}(#1,#2)}}
\newcommand{\snapdown}[2]{\ensuremath{\mathrm{down}(#1,#2)}}
\begin{proof}

For a node $v$, consider all embeddings of $T(v)$ and the longest decreasing subsequence each implies: let $f(v,a,b)$ be the minimum length of a longest decreasing subsequence taken over the permutations $\Pi(T(v))$, where any values outside of $[a,b]$ are ignored for the subsequence.
Then by definition $f(\text{root}(T),1,n)$ solves the full problem.
Note that $f$ has optimal substructure.
The optimum is witnessed by an embedding $\pi$ of the tree, and $\pi$ implies a range of integers for every subtree that the longest decreasing subsequence conforms to;
if any subtree has a strictly better solution in a particular range, it can be substituted into the purportedly optimal solution.

Abusing notation to write the leaf itself to mean its integer value, observe that a leaf $\ell$ trivially has $f(\ell,a,b)=1$ if $\ell\in [a,b]$, and $0$ otherwise.
For an internal node $v$, any decreasing subsequence in an embedding of $T(v)$ starts in the subtree that is embedded to the left, and continues in the other subtree using only smaller values.
Therefore
\begin{equation}\label{eq:lds-dp}
    f(v,a,b) = \min\Big\{\ \max_{s\in[a,b]} f(c_1,a,s) + f(c_2,s, b),\quad \max_{s\in[a,b]} f(c_2,a,s) + f(c_1,s, b)\ \Big\}.
\end{equation}
A straightforward analysis of computing $f(\text{root}(T),1,n)$ by dynamic programming gives a runtime bound of \BigO{n^4}, since there are \BigO{n^3} distinct combinations of arguments, and $f$ can be evaluated in \BigO{n} time. %
However,
using some structural properties of $f$ the maximum in Equation~\eqref{eq:lds-dp} can be evaluated by looking at only \BigO{\min\{n(c_1),n(c_2)\}} terms.
In particular, look at the case of putting $c_1$ to the left; the other case is symmetric.
Then we evaluate
\begin{equation}\label{eq:one-max}
    \max_{s\in[a,b]} f(c_1,a,s) + f(c_2,s, b).
\end{equation}
First, note that we include $s$ in the range for both recursive calls, but this is harmless since the leaf $s$ occurs in at most one of the subtrees and thus cannot be double counted.
Now consider the case where $a$ or $b$ does not occur in $T(v)$: we can snap it up (resp. down) to the first integer that does, without changing the function value -- in terms of values that could go in decreasing subsequences, the situation described is equivalent.

\begin{obclaim}
    Let $\snapup{v}{i}=\min\{\ell\in L(v) \colon \ell\geq i\}$; $\snapdown{v}{i}=\max\{\ell\in L(v) \colon \ell\leq i\}$.
    Then $f(v,a,b)=f(v,\snapup{v}{a},\snapdown{v}{b})$.
\end{obclaim}
This can be used in practice to reduce the number of DP states being evaluated by forwarding to the ``snapped'' values of $a$ and $b$ in constant time (by precomputing \snapup{\cdot}{\cdot} and \snapdown{\cdot}{\cdot}).
For a worst case bound, we further observe about \Cref{eq:one-max} that the maximum must be achieved by a value of $s$ that is snapped to $c_1$, also one snapped to $c_2$.
\begin{claim}\label{lem:check-smaller}
    The maximum in \Cref{eq:one-max} is achieved by some $s_1\in L(c_1)$; it is also achieved by some $s_2$ with $s_2+1\in L(c_2)$. (Possibly $s_1=s_2$.)
\end{claim}
\begin{claimproof}
    For fixed $v$, $f(v,a,b)$ is increasing in $b$ (more options), only changing when $b\in L(v)$; conversely, it is decreasing in $a$ (fewer options), only changing one step after any $a\in L(v)$.
    The terms in \Cref{eq:one-max} are thus a sum of terms from an increasing and a decreasing sequence.
    Consider any interval of values for $s$  such that $f(c_1,a,s)+f(c_2,s,b)$ that realizes the maximum: it starts when a new value in $L(c_1)$ is added, and it ends one before a value in $L(c_2)$ is removed.
    Hence the maximum is witnessed at a value in the left subtree, and (possible for a different value of $s$) also one before a value in the right subtree.
\end{claimproof}
\Cref{lem:check-smaller} lets us evaluate the maxima in \Cref{eq:lds-dp} either by checking the values in $L(c_1)$ \emph{or} the values in $L(c_2)$ minus one -- our algorithm will pick the smaller of the two; this works independently for both maxima.
The time spent to evaluate a particular $f(v,a,b)$ is thus \BigO{\min\{n(c_1),n(c_2)\}}.
Then for fixed $a$ and $b$, aggregating over the entire tree, only \BigO{n\log n} terms are evaluated in the maxima (instead of \BigO{n^2}).
This can be seen by charging the terms to the leaves in the \emph{smaller} subtree (of which there are equally many): a leaf can be ``on the smaller side'' only \BigO{\log n} times, hence the total charge is \BigO{n\log n}.
Summing this over the \BigO{n^2} combinations of $a$ and $b$ yields the theorem.
\end{proof}

\section{Tree-Consistent Colorings}
\label{sec:tree-consistent}

We now consider tree-consistent colorings and show in \Cref{sec:tree-consistent-dp} that we can find a minimum tree-consistent colored embedding in polynomial time.
As a subroutine, we determine whether a 
geophylogeny %
admits a planar labeling.
Klawitter et al.~\cite{KKS+.VGI.2025} show that this can be done in \BigO{n^6} time by dynamic programming over the embedding of $T$.
In the next section, we first present a different DP approach that runs in $\BigO{n^5}$ time.

\subsection{A Faster Algorithm for Geophylogeny Planarity Testing}
\label{sec:planar}
We employ a DP that relies on the following known observation%
~\cite{DNTW.Cbl.2025,BHKN.AMC.2009}; see also \shortLong{\Cref{fig:planarity-merged}a}{\Cref{fig:planarity-split}}.

\shortLong
{
\begin{figure}[t]
	\centering
	\includegraphics[page=1]{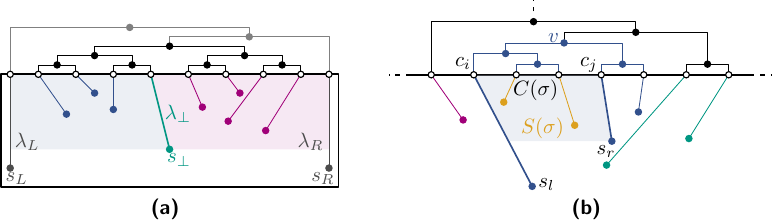}
	\caption{\textbf{\textsf{(a)}} The leader $\leader_{\bot}$ for the bottommost site $s_{\bot}$ splits the labeling \labeling into the blue and the purple sub-labelings. We indicate the sentinel leaves in gray.
    \textbf{\textsf{(b)}} The state $\sigma = (s_l, c_i, s_r, c_j)$ (blue) and the sets $S(\sigma)$ and $C(\sigma)$. For every orange site $s_p \in S(\sigma)$ we have $\ell_p \in L(T(v))$, $v = \lca{\ell_l}{\ell_r}$.
    }
	\label{fig:planarity-merged}
\end{figure}
}
{
\begin{figure}
	\centering
	\includegraphics[page=1]{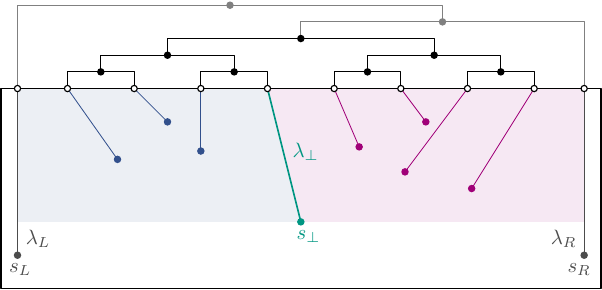}
	\caption{The leader $\leader_{\bot}$ for the bottommost site $s_{\bot}$ splits the labeling \labeling into the blue and the purple sub-labeling. We indicate the sentinel leaves in gray.}
	\label{fig:planarity-split}
\end{figure}
}

\begin{observation}
    \label{obs:planar-bottommost}
    The leader $\leader_{\bot} = (s_{\bot}, c_{p})$ for the bottommost site $s_{\bot} \in S$ splits a planar labeling \labeling into two independent sub-labelings for the sites left and right of $\leader_{\bot}$, respectively.
\end{observation}
This allows us to describe sub-labelings of \labeling via two leaders $\leader_l = (s_l, c_i)$ and $\leader_r = (s_r, c_j)$, $i < j$.
The entire labeling \labeling can be described by adding two sentinel leaves $\ell_L$ and $\ell_{R}$ to~$T$ and two dummy candidates $c_L$ and $c_{R}$ to $C$ that lie at the very left and right side of~$R$.
The respective sites $s_L$ and $s_{R}$ lie %
at the bottom corners of $R$ and the dummy leaders $\leader_L = (s_L, c_L)$ and $\leader_R = (s_{R}, c_{R})$ bound $G$; recall \shortLong{\Cref{fig:planarity-merged}a}{\Cref{fig:planarity-split}}.
\shortLong{
The corresponding scheme of repeatedly finding a candidate for the bottommost site is known in the literature~\cite{DNTW.Cbl.2025,BNN.ELF.2021}, and we now sketch how to use it for geophylogenies; see %
\NewText{the full version~\cite{ARXIV}} for details.
}{
This gives rise to an alternative approach for planarity testing of geophylogenies, namely via a DP over sub-labelings.
Note that a similar approach has already been used by Depian et al.~\cite{DNTW.Cbl.2025} to obtain a boundary labeling that adheres to grouping constraints,which can also be represented as a (not-necessarily binary) tree on the sites.
However, in contrast to Depian et al., who ensured admissibility when evaluating the candidate leaders for the bottommost site, we argue in this section that for geophylogenies, it sufficient to check for admissibility of sub-labelings only.
This results in a single check for each sub-labeling compared to one check per candidate leader (which could be \BigO{n}-many).
}
\onlyShort{%

\subparagraph*{A Glimpse at the DP.}
We represent a DP-\emph{state} $\sigma = (s_l, c_i, s_r, c_j)$ via two leaders $\leader_l = (s_l, c_i)$ and $\leader_r = (s_r, c_j)$, $i < j$, and let $S(\sigma)$ and $C(\sigma)$ denote the sites and candidates between the two leaders, respectively; see also \Cref{fig:planarity-merged}b. 
Our DP maintains $\BigO{n^4}$-sized DP-table $D$, which stores for every state $\sigma = (s_l, c_i, s_r, c_j)$ whether it is \emph{valid}, i.e., $\Size{S(\sigma)} = \Size{C(\sigma)}$ holds and there exists an embedding $\pi\in \Pi(T)$ that induces a planar labeling $\labeling(\sigma)\colon S(\sigma)\cup \{s_l, s_r\} \to C(\sigma) \cup \{c_i, c_j\}$ such that $\labeling(s_l) = c_i$ and $\labeling(s_r) = c_j$.
We observe that for every valid state~$\sigma$ and every $s_p \in S(\sigma)$, we have $\ell_p \in L(T(v))$, i.e., only sites of $T(v)$ can be labeled at $C(\sigma)$; see the orange sites in \Cref{fig:planarity-merged}b.
This gives rise to the concept of a \emph{consistent} state, where (i) $\Size{S(\sigma)} = \Size{C(\sigma)}$ holds, (ii) $\leader_l$ and $\leader_r$ do not cross, (iii) every site $s_i \in S(\sigma)$ satisfies $\ell_i \in L(T(\lca{\ell_1}{\ell_2}))$, and (iv) there exists an embedding $\pi \in \Pi(T)$ with $\pi(\ell_l) = i$, $\pi(\ell_r) = j$ such that for every $\ell_p \in L(T(\lca{\ell_l}{\ell_r}))$ whose site $s_p$ lies above $s_l$ and $s_r$ we have $\pi(\ell_p) < i$, $i < \pi(\ell_p) < j$, or $\pi(\ell_p)$ according to whether $s_p$ lies left of $\leader_l$, between $\leader_l$ and $\leader_r$, or right of $\leader_r$, respectively.
}%
\begin{statelater}{planarSetup}
\subparagraph*{Setting up the DP.}
Let $s_l, s_r \in S$ be two sites and $c_i, c_j \in C$ be two candidates with $i < j$.
The tuple $\sigma = (s_l, c_i, s_r, c_r)$ will denote a \emph{state} of our DP.
Observe that the leaders $\leader_l = (s_l, c_i)$ and $\leader_r = (s_r, c_j)$ together with the top boundary of $R$ induce a trapezoid $\tau(\sigma)$ in $R$.
We let $S(\sigma)$ and $C(\sigma)$ denote the sites and candidates \emph{contained} in $\sigma$, i.e., that lie in $\tau(\sigma)$, excluding $s_l$, $s_r$, $c_i$, and $c_j$; see \shortLong{\Cref{fig:planarity-merged}b}{\Cref{fig:planarity-lca}} for an example.
The state $\sigma$ is \emph{valid} if $\Size{S(\sigma)} = \Size{C(\sigma)}$ holds and there exists an embedding $\pi\in \Pi(T)$ that induces a planar labeling $\labeling(\sigma)\colon S(\sigma)\cup \{s_l, s_r\} \to C(\sigma) \cup \{c_i, c_j\}$ such that $\labeling(s_l) = c_i$ and $\labeling(s_r) = c_j$.
In essence, for $\labeling(\sigma)$ we only care about %
the sites in $S(\sigma)\cup \{s_l, s_r\}$.
In particular, we do not enforce planarity of the entire labeling $\labeling'$ for $S$ induced by the permutation $\pi$.
We now make a crucial observation about valid states; see also \shortLong{\Cref{fig:planarity-merged}b}{\Cref{fig:planarity-lca}}.

\onlyLong{\begin{figure}
	\centering
	\includegraphics[page=1]{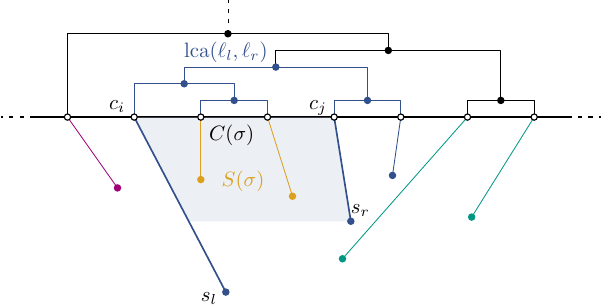}
	\caption{The state $\sigma = (s_l, c_i, s_r, c_j)$ (indicated in blue) and the sets $S(\sigma)$ and $C(\sigma)$. For every orange site $s_p \in S(\sigma)$ we have that $\ell_p \in L(T(\lca{\ell_l}{\ell_r}))$ (blue).}
	\label{fig:planarity-lca}
\end{figure}
}

\begin{observation}
    \label{obs:planar-valid-state}
    Let $\sigma = (s_l, c_i, s_r, c_j)$ be a valid state and $v = \lca{\ell_l}{\ell_r}$.
    For every $s_p \in S(\sigma)$ we have $\ell_p \in L(T(v))$, i.e., only sites of $T(v)$ can be labeled at $C(\sigma)$.
    In particular, we have $\ell_l, \ell_r \in L(T(v))$.
\end{observation}
\Cref{obs:planar-valid-state} gives rise to concept of %
\emph{consistent} states.
A state $\sigma = (s_l, c_i, s_r, c_j)$ is consistent if and only if all of the following conditions hold, where $\leader_l = (s_l, c_i)$ and $\leader_r = (s_r, c_j)$.

\begin{enumerate}[(Con~1)]
	\item We have $\Size{S(\sigma)} = \Size{C(\sigma)}$. \label[Con]{planar:same-cardinality}
	\item The leaders $\leader_l$ and $\leader_r$ do not cross. \label[Con]{planar:crossing-free}
	\item The state $\sigma$ contains no \emph{trapped} site, i.e., no site $s_i \in S(\sigma)$ with $\ell_i \notin L(T(\lca{\ell_1}{\ell_2}))$.%
    \label[Con]{planar:not-trapped}
	\item There exists an embedding $\pi \in \Pi(T)$ of $T$ such that (i) $\pi(\ell_l) = i$, (ii) $\pi(\ell_r) = j$, and (iii) for every leaf $\ell_p \in L(T(\lca{\ell_l}{\ell_r}))$ such that $s_p$ is above $s_l$ and $s_r$ we have 
	\begin{itemize}
		\item if $s_p$ is left of $\leader_l$ then $\pi(\ell_p) < i$,
		\item if $s_p$ is between $\leader_l$ and $\leader_r$ then $i < \pi(\ell_p) < j$, and
		\item if $s_p$ right of $\leader_r$ then $j <\pi(\ell_p)$.
	\end{itemize}
	\label[Con]{planar:embedding}
\end{enumerate}

Observe that \Cref{planar:same-cardinality,planar:crossing-free,planar:not-trapped,planar:embedding} are necessary for the validity of $\sigma$.
\end{statelater}%
\onlyLong{Furthermore, note that for \Cref{planar:embedding}, we only care about the leaves for which the corresponding site lies above the two sites that define the state.
For all other leaves, some leaf in the same subtree or the position of the subtree on the path to $\lca{\ell_l}{\ell_r}$ will dictate their placement in $\pi$.
Moreover, we will take care of these sites in some other (larger) state $\sigma'$.} %
\onlyLong{We need one more notion before we can present our dynamic program.
Finally, let $s_{\bot} \in S(\sigma)$ be the bottommost site in the state~$\sigma$.
A candidate $c_p \in C(\sigma)$ is \emph{feasible} for $s_{\bot}$ if the leader $\leader_{\bot} = (s_{\bot}, c_p)$ does not intersect another site $s_q \neq s_p$.}
\shortLong{ 
To evaluate a DP-state $\sigma$, we check consistency and, if so, whether a candidate $c_p \in C(\sigma)$ is \emph{feasible} for the bottommost site $s_{\bot} \in S(\sigma)$, i.e., the leader $(s_{\bot}, c_p)$ does not intersect any other site, and we have $D[(s_l, c_i, s_{\bot}, c_p)] = D[(s_{\bot}, c_p, s_r, c_j)] = 1$
(if $S(\sigma) = \emptyset$, being consistent is sufficient for $D[\sigma] = 1$).

Verifying Condition~(iv) dominates the consistency checks, since we can precompute lowest common ancestor queries~\cite{GT.LTA.1985} and leader-site intersections. 
To check it in linear time, we compute via a postorder traversal of $T$ the numbers $S(\text{left},u)$, $S(\text{in}, u)$, and $S(\text{right}, u)$ of leaves of $T(u)$ whose sites lie left of $\leader_l$, between $\leader_l$ and $\leader_r$, and right of $\leader_r$, respectively.
Afterwards, we traverse $T$ from $\ell_l$ to $v = \lca{\ell_l}{\ell_r}$ (and symmetrically for $\ell_r$ to $v$).
We maintain counters of leaves which must be placed left of $\ell_l$, between $\ell_l$ and $\ell_r$, and right of~$\ell_r$, which we can determine based on the counter $S(\cdot, w)$ for every node $w$ being a child of a node on the traversal but not on the traversal themselves; see \Cref{fig:planarity-inside-outside-detailed}.
After both traversals, we reject if any counter is larger than the actual available candidates, and verify that there exists an embedding of $T$ that places $T(v)$ at the computed leftmost position.
The existence of such embeddings can be precomputed via a top-down DP of $T$.
We obtain:
\begin{figure}[t]
    \centering
    \includegraphics[page=1]{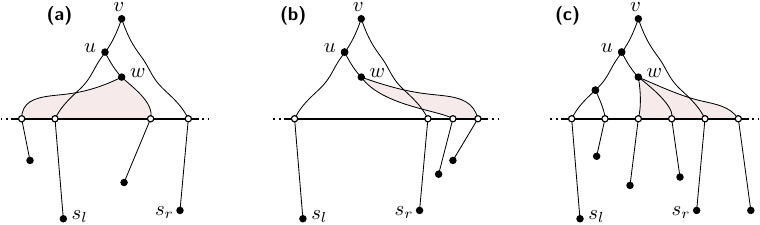}
    \caption{Let $u$ be a node on the path from $\ell_l$ to the lowest common ancestor and $w$ the child which is not on the path.
    Each of the figures shows a setting in which we detect that $\sigma$ is inconsistent since the desired embedding $\pi$ does not exist as indicated with the shaded subtree:
    \textbf{\textsf{(a)}} If $S(\text{in}, w) > 0$ and $S(\text{left}, w) + S(\text{right}, w) > 0$, then $T(w)$ simultaneously contains leaves for sites that are inside and outside $\sigma$.
    \textbf{\textsf{(b)}} If $S(\text{right}, w) > 0$, then $T(w)$ contains a leaf for a site right of $\sigma$; however, we are on a traversal from $\ell_l$ to $v$.
    \textbf{\textsf{(c)}} If $S(\text{in}, w) \neq n(w)$, then only some sites for the leaves of $T(w)$ are $S(\sigma)$ (the others are outside but, for example, below $s_l$ or $s_r$).}
    \label{fig:planarity-inside-outside-detailed}
\end{figure}
}{

We are now ready to define our DP, which maintains a $\BigO{n^4}$-sized DP-table $D$.
For a state $\sigma = (s_l, c_i, s_r, c_j)$, we set $D[\sigma] = 1$ if and only if $\sigma$ is consistent and there exists a feasible candidate $c_p \in C(\sigma)$ for the bottommost site $s_{\bot} \in S(\sigma)$ such that $D[(s_l, c_i, s_{\bot}, c_p)] = D[(s_{\bot}, c_p, s_r, c_j)] = 1$.
Of course, if $S(\sigma) = \emptyset$, then being consistent is sufficient for $D[\sigma] = 1$.
In all other cases, we set $D[\sigma] = 0$.
It remains to show correctness and running time.
}

\begin{statelater}{planarCorrectness}
\begin{lemma}
	\label{lem:planar-correct}
	For every state $\sigma = (s_l, c_i, s_r, c_j)$ it holds $D[\sigma] = 1$ if and only if $\sigma$ is valid.
\end{lemma}
\begin{proof}
	Let $\instance = \instancelong$ be a geophylogeny and $\sigma = (s_l, c_i, s_r, c_j)$ a state.
	For the backward direction $(\Leftarrow)$, assume that there exists an embedding $\pi \in \Pi(T)$ that induces a planar labeling $\labeling(\sigma)$ witnessing the validity of $\sigma$.
	To see that $D[\sigma] = 1$ must hold, it is sufficient to observe that the existence of $\labeling(\sigma)$ also witnesses (i) consistency (every planar labeling witnesses \Cref{planar:same-cardinality,planar:crossing-free,planar:not-trapped,planar:embedding}) and (ii) the existence of a feasible candidate for the bottommost site $s_{\bot}$ for $\sigma$ and all corresponding sub-states.
	Therefore, we focus on the more interesting forward direction $(\Rightarrow)$.
	
	To show this direction, we use induction on the size of the state, i.e., $\Size{S(\sigma)}$.
	For the base case ($\Size{S(\sigma)} = 0$), note that $D[\sigma] = 1$ simply means that $\sigma$ is consistent.
	\Cref{planar:same-cardinality,planar:crossing-free,planar:embedding} together ensure that $\sigma$ is valid.
	In particular, the embedding whose existence is guaranteed by \Cref{planar:embedding} induces a planar labeling due to \Cref{planar:crossing-free}.
	
	For the inductive step, consider the more general case with $\Size{S(\sigma)} > 0$.
	Since $D[\sigma] = 1$, we have $\Size{S(\sigma)} = \Size{C(\sigma)}$ and there exists for the lowestmost site $s_{\bot} \in S(\sigma)$ a candidate $c_p \in C(\sigma)$ such that we have $D[\sigma_1] = D[\sigma_2] = 1$ for the two resulting sub-states $\sigma_1 = (s_l, c_i, s_{\bot}, c_p)$ and $\sigma_2 = (s_{\bot}, c_p, s_r, c_j)$.
	Note that we have $\Size{S(\sigma_1)}, \Size{S(\sigma_2)} < \Size{S(\sigma)}$.
	Therefore, by induction hypothesis, there exists two embeddings $\pi_1$ and $\pi_2$ from $\Pi(T)$ that witness validity of $\sigma_1$ and~$\sigma_2$, respectively.
	Since $\pi_1$ and $\pi_2$ witness validity, they induce a planar labeling $\labeling_1$ and $\labeling_2$ of the state.
	Let $v = \lca{\ell_l}{\ell_r}$, $v_1 = \lca{\ell_l}{\ell_{\bot}}$, and $v_2 = \lca{\ell_{\bot}}{\ell_r}$.
	Since $\sigma$ is consistent, we have $\ell_l, \ell_r, \ell_{\bot} \in L(T(v))$.
	Therefore, either $v_1 = v$ or $v_2 = v$.
	We assume without loss of generality the former case; all of the following arguments are symmetric for the latter case.
	
	We now show that we can combine $\labeling_1$ and $\labeling_2$ into a valid labeling of $\sigma$ and use it show that there exists an embedding $\pi\in\Pi(T)$ that witnesses validity of $\sigma$
	Let $\labeling\colon S(\sigma) \cup \{s_l, s_r\} \to C(\sigma) \cup \{c_i, c_j\}$ be a labeling defined as follows.
	Set $\labeling(s_l) = c_i$, $\labeling(s_{\bot}) = c_p$, and for every site $s_q \in S(\sigma)$ set $\labeling(s_q) = \labeling_1(s_q)$ if $s_q \in S(\sigma_1)$, $\labeling(s_q) = \labeling_2(s_q)$ if $s_q \in S(\sigma_2)$, and $\labeling(s_q) = c_j$ if $s_q = s_{\bot}$.
	We can observe that $\labeling$ must be planar since $\pi_1$ and $\pi_2$ induce planar labelings which are ``separated'' by the leader $(s_{\bot}, c_p)$.
	To see that it also respects the constraints induced by $T$, we assume it would not, i.e., there exists three leaves $\ell_x, \ell_y, \ell_z$ such that $\pi(\ell_x) < \pi(\ell_y) < \pi(\ell_z)$ and $\ell_z \notin u = L(T(\lca{\ell_x}{\ell_z}))$.
	We must have that $v$ is ancestor of $u$, as otherwise $s_y$ would be trapped in $\sigma$.
	Furthermore, observe that $s_x$, $s_y$, and $s_z$ cannot be lower than $s_{\bot}$.
	Clearly, $s_x$, $s_y$, and $s_z$ cannot be in the same sub-state since this would contradict validity of $\sigma_1$ or $\sigma_2$.
	Thus, from $\pi(\ell_x) < \pi(\ell_y) < \pi(\ell_z)$, we get that $s_x \in S(\sigma_1) \cup \{s_l\}$ and $s_z \in S(\sigma_2) \cup \{s_r\}$ must hold.
	Note that neither of $s_x$ and $s_z$ can be~$s_{\bot}$, as then all three sites would be part of the same sub-state, which we have already ruled out.
	We now show that any choice of site for $s_x$ leads to a contradiction; consider \Cref{fig:planarity-correctness} for the following arguments.
    \begin{figure}
    	\centering
    	\includegraphics[page=1]{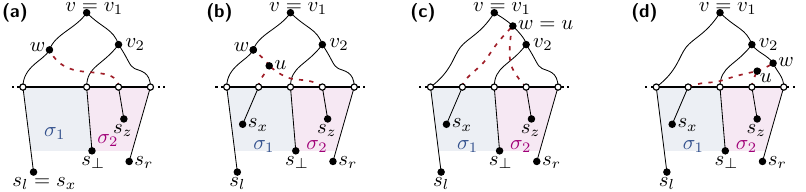}
    	\caption{Visualization of the arguments why combining a labelings for $\sigma_1$ and $\sigma_2$ results in an admissible labeling. \textbf{\textsf(a)} shows the case with $s_x = s_l$ and \textbf{\textsf{(b)}}--\textbf{\textsf{(c)}} show the different cases for $s_x \in S(\sigma_1)$.}
    	\label{fig:planarity-correctness}
    \end{figure}
	
	First, $s_x = s_l$ is not possible:
	If $s_z = s_r$, then $s_y$ is trapped. 
    If $s_z \in S(\sigma_2)$, we %
    observe that the first ancestor $w$ of $u$ is on a path from $\ell_l$ %
    to $v$ (if $u$ itself is on such a path, then $w = u$).
    Thus, $\sigma_1$ contains a subtree with a site right of the state (namely the site $s_z$, which is then above $s_l$ and $s_{\bot}$); see also \Cref{fig:planarity-correctness}a.

	Second, $s_x \in S(\sigma_1)$ is not possible:
	We consider the location of the first ancestor $w$ of $u$ on a path from $\ell_l$, $\ell_{\bot}$, or $\ell_r$ to $v$ (or $w = u$ if $u$ is already on such a path).
	If $w \neq u$ and is on the path from $\ell_l$ or $\ell_{\bot}$ to $v$, then $\sigma_1$ is inconsistent because it contains a subtree with a site right of the state (namely the site $s_z$); see \Cref{fig:planarity-correctness}b.
	If $w = u$, then we can conclude that $\labeling_2$ already did not respect the constraints with similar arguments as above (either because~$s_z$ is trapped or because $s_x$ is left of the state as in \Cref{fig:planarity-correctness}c).
	Otherwise, $w \neq u$ is on the path from $\ell_r$ to $v_2$, which implies that $\sigma_2$ is inconsistent because it contains a subtree with a site left of the state (namely the site $s_x$); see \Cref{fig:planarity-correctness}d.
	All cases lead to a contradiction with the assumed validity of $\sigma_1$ and $\sigma_2$ or consistency of $\sigma$.
	Thus, the labeling $\labeling$ respects the constraints imposed by $T$.
	In particular, it witnesses that there exists an embedding~$\pi'$ of $T(v)$ where $\pi'(\ell_l) = i$, $\pi'(\ell_r) = j$, and $i < \pi'(\ell_q) < j$ for every $s_q \in S(\sigma)$.
	\Cref{planar:embedding} now guarantees that we can ``extend'' $\pi'$ to an embedding $\pi \in \Pi(T)$ of $T$,
	In particular, $\pi'$ just specifies the embedding of the leaves for sites in $S(\sigma)$ (and $s_l$ and $s_r$) but maintain the interface to the remaining leaves via $\pi'(\ell_l) = i$ and $\pi'(\ell_r) = j$.
	Combining all, we conclude that $\sigma$ is valid.
\end{proof}
\end{statelater}

\onlyLong{
Before we use \Cref{lem:planar-correct} to establish \Cref{thm:planar}, we first describe how to implement our DP efficiently.
}

\begin{statelater}{planarEfficient}
\subparagraph*{An Efficient Implementation.}
While the previous paragraph gives us all ingredients for the DP, implemented na{\"i}vely, this will not be faster than the existing algorithm by Klawitter et al.~\cite{KKS+.VGI.2025}.
To this end, we make use several data structures and lookup tables that enable us to evaluate a state $\sigma$ in $\BigO{n}$ time.
First, we use an efficient data structure to query $\lca{\ell_i}{\ell_j}$ in constant time (see e.g.\ Gabow and Tarjan~\cite{GT.LTA.1985}).
Next, we compute for every site $s_i$ and candidate $c_j$ in $\BigO{n}$ time whether the leader $(s_i,c_j)$ intersect another site $s_p \neq s_i$.
Finally, we need to efficiently determine whether there exists an embedding $\pi \in \Pi(T)$ that places a subtree $T(v)$ for $v\in V$ starting at position $i$.
To this end, we precompute a map $\text{con}(v, i)$, which stores whether there exists an embedding $\pi$ such that $i \leq \pi(\ell_j) \leq i + n(v) - 1$ for every $\ell_j \in L(T(v))$ as follows.
We maintain a queue $Q$, which we initialize with $(v, 0)$.
While~$Q$ is non-empty, we (i) remove the first entry from $(v, i) \in Q$, (2) check if $\text{con}(v, i)$ is already set, and if not, (3) set $\text{con}(v, i) = 1$ and add $(x, i)$, $(x, i + n(y))$, $(y, i)$, and $(y, i + n(x))$ to $Q$, for $x, y \in \children{v}$ (if they exist), which corresponds to the two possible orientations around $v$.
Since every tuple $(v,i)$ is in $Q$ at most once, this takes $\BigO{n^2}$ time overall.
In the end, we set $\text{con}(v, i) = 0$ for every $v \in V$, $i \in [0]$ where $\text{con}(v, i)$ has not been set to $1$, which takes $\BigO{n^2}$ time.
\end{statelater}

\onlyLong{
Using these data structures, our dynamic program can be implemented to run in $\BigO{n^5}$ time as follows.
}

\begin{restatable}\restateref{thm:planar}{theorem}{theoremPlanar}
	\label{thm:planar}
	Let \instance be a geophylogeny on $n$ taxa.
    We can compute a planar labeling of \instance in $\BigO{n^5}$ time and $\BigO{n^4}$ space or decide that no such labeling exists.
\end{restatable}
\begin{prooflater}{ptheoremPlanar}
	Let $G$ be a geophylogeny on $n$ taxa.
	We add the sentinel leaves $\ell_L$ and $\ell_R$ and their corresponding dummy sites $s_L, s_R$ and candidates $c_L$ and $c_R$, respectively.
	Afterwards, we fill the DP-table $D$.
	By repeated application of \Cref{lem:planar-correct}, we conclude that $G$ admits a planar labeling if and only if $D[(s_L, c_L, s_R, c_R)] = 1$.
	It remains to argue running time.
	To this end, observe that each lookup-table takes $\BigO{n^2}$ space and can be computed in $\BigO{n^3}$ time.
	Moreover, there are $\BigO{n^4}$ possible states $\sigma$.
	Let $\sigma = (s_l, c_i, s_r, c_j)$ be one of these states.
	We now argue that we can evaluate it in linear time provided that all states $\sigma'$ with $\Size{S(\sigma')} < \Size{S(\sigma)}$ are already evaluated.
	
	When evaluating $\sigma$, we first have to compute the sets $S(\sigma)$, which we can do by iterating through all sites $S$ and checking whether they are above $s_l$ and $s_r$ and between the leaders $(s_l, c_i)$ and $(s_r, c_j)$.
	To determine the set $C(\sigma)$, we iterate through all candidates $c_p$ and check if $i < p < j$.
	Both steps take linear time and we can also determine the bottommost site $s_{\bot} \in S(\sigma)$ while computing $S(\sigma)$.
	
	In order to check if $\sigma$ is consistent, we let $v = \lca{\ell_i}{\ell_j}$ be the lowest common ancestor of the two leaves whose sites define $\sigma$.
	Obtaining $v$ takes constant time since we have computed the lowest common ancestor of every leaf pair.
	For each internal node $u \in I(T(v))$, we let $S(\text{left},u)$ denote the number of leaves $\ell_p \in L(T(u))$ such that $s_p$ is above $s_i$ and $s_j$ and left of the leader $(s_l, c_i)$.
	The numbers $S(\text{in}, u)$ and $S(\text{right}, u)$ are defined analogously but with respect to being between $(s_l, c_i)$ and $(s_r, c_j)$ and to the right of $(s_r, c_j)$, respectively.
	We can obtain all of these numbers via a single postorder traversal of $T(v)$.
	
	To check if $\sigma$ is consistent, we have to evaluate each of the four consistency-criteria \Cref{planar:same-cardinality,planar:crossing-free,planar:not-trapped,planar:embedding}.
	\Cref{planar:same-cardinality,planar:crossing-free} can easily be checked in constant time.
	For \Cref{planar:not-trapped}, i.e., whether there is no trapped site, we compare $\Size{S(\sigma)}$ with $\Size{S(\text{in}, v)}$ and reject the state if they are different as this implies that there is a site in $S(\sigma)$ which does not belong to the subtree $T(v)$.
	It remains to check \Cref{planar:embedding}, i.e., the existence of a suitable embedding $\pi \in \Pi(T)$ and for this, we initialize three ``global'' counters $C_{\text{left}} = i - 1$, $C_{\text{in}} = j - i - 1$, and $C_{\text{right}} = n - j$ with the number of candidates to the left of $c_i$, between $c_i$ and $c_j$, and to the right of $c_j$, respectively.
	We then traverse the tree $T$ from $\ell_l$ to $v$ and do the following at each internal node $u \neq v$ that we encounter.
	Let $w \in \children{u}$ be the child that does not contain $\ell_l \in L(T(w))$.
	We now make the following checks in the following order; see also \Cref{fig:planarity-inside-outside} for a visualization.
    \begin{figure}
    	\centering
    	\includegraphics[page=1]{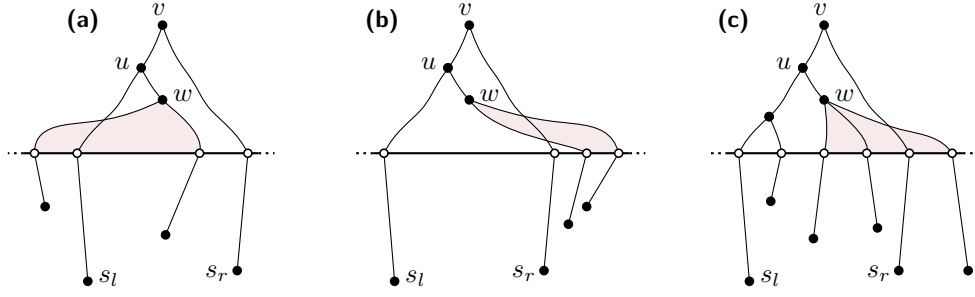}
    	\caption{Let $u$ be a node on the path from $\ell_l$ to the lowest common ancestor and $w$ the child which is not on the path.
        We detect that the state is inconsistent (indicated with the red subtree) because \textbf{\textsf{(a)}} $S(\text{in}, w) > 0$ and $S(\text{left}, w) + S(\text{right}, w)$ > 0, \textbf{\textsf{(a)}} $S(\text{right}, w) > 0$, and \textbf{\textsf{(c)}}, only some the sites $L(T(w))$ are in $S(\sigma)$ (but none of them is in $S(\text{left}, w) + S(\text{right}, w)$.}
    	\label{fig:planarity-inside-outside}
    \end{figure}
    
	\begin{enumerate}
		\item If $S(\text{in}, w) > 0$ and $S(\text{left}, w) + S(\text{right}, w) > 0$, then $w$ contains leaves for sites that are in the state and outside the state.
		As this cannot lead to a desired embedding, we reject the state; see also \Cref{fig:planarity-inside-outside}a.
		\item If $S(\text{right}, w) > 0$, then $w$ contains a leaf whose site sits to the right of the state; see also \Cref{fig:planarity-inside-outside}b.
		However, since we are on a traversal from  $\ell_l$ to $v$, there cannot exist a valid embedding that places that leaf at a position $q > j$.
		Thus, we reject the state.
		\item If $S(\text{in}, w) > 0$, then we know that $w$ contains a leaf whose site is in the state.
		This implies that all leaves $L(T(w))$ be placed somewhere between position $i$ and $j$.
		Thus, we check if $S(\text{in}, w) = n(w)$, reject if this is not the case as in \Cref{fig:planarity-inside-outside}c, and decrease $C_{\text{in}}$ by~$n(w)$.
		\item If we did not reject the state until now, we know that $S(\text{in}, w) = 0$ holds.
		Thus, $L(T(w))$ must be placed somewhere before position $i$ and we decrease $C_{\text{left}}$ by $n(w)$.
	\end{enumerate}
	Once we are at $v$, we perform symmetrically for $\ell_r$.
	Note that all of the above embedding choices are completely determined by the fact that we must have $\pi(\ell_l) = i$ and $\pi(\ell_r) = j$.
	We reject the state if $C_{\text{left}} < 0$, $C_{\text{in}} \neq 0$, or $C_{\text{right}} < 0$ holds, since this implies an insufficiency (or mismatch) of the candidates and thus the non-existence of the sought-after embedding~$\pi$.
	Up until now, we have checked if $T(v)$ itself admits the desired embedding.
	However, it remains to check if this embedding is compatible with the remainder of the tree $T$.
	To this end, observe that the value of $C_{\text{left}}$ tells us the number of leaves that must come before the first leaf of $T(v)$ in the embedding $\pi$.
	Or, viewed differently, in $\pi$, the first leaf of $T(v)$ must be placed at position $C_{\text{left}} + 1$.
	Thus, we check if $\text{con}(v, C_{\text{left}} + 1) = 1$ holds and reject the state otherwise.
	If none of the above checks failed, then we know that $\sigma$ is consistent.
	Overall, this constant time per internal node on the two paths to $v$.
	Since there are $\BigO{n}$ internal nodes on each of the two paths, consistency can be checked in linear time.

	Finally, we have to check for each of the $\BigO{n}$ candidates $c_p \in C(\sigma)$ if they are feasible and, if so, evaluate the corresponding sub-states.
	Since we have precomputed if $(s_{\bot}, c_p)$ is feasible, this takes only constant time per candidate.
	Overall, we conclude that we can evaluate a state $\sigma$ in $\BigO{n}$ time.
	Since there are $\BigO{n^4}$ states, we can fill the entire DP table in 
	$\BigO{n^5}$ time.
	Obtaining a witnessing planar labeling (if there exists one), can be obtained in the same time via standard backtracking techniques and the statement thus follows.
\end{prooflater}
When evaluating a state $\sigma$, we can also optimize a so-called \emph{leaf-additive} objective function\footnote{A leaf-additive function can be expressed as $\sum_{s \in S} f(\labeling(s))$, i.e., only depends on the label position for each site but ignores the rest of the labeling~\cite{KKS+.VGI.2025}.},
which includes common objectives such as the length of the leaders~\cite{BHKN.AMC.2009}.%
\begin{corollary}
	\label{cor:planar-leaf-additivie}
	Let $G$ be a geophylogeny $G$ on $n$ taxa and let $f$ be a leaf-additive optimization function.
	In $\BigO{n^5}$ time and $\BigO{n^4}$ space we can compute a planar labeling of $G$ that minimizes/maximizes $f$ or conclude that no planar labeling exists.
\end{corollary}

\subsection{Tree-Consistent Coloring in Polynomial Time}
\label{sec:tree-consistent-dp}
We design a dynamic program that builds upon the planarity test from \Cref{sec:planar}.
In particular, observe that for an interior node $v$, if $T(v)$ admits a planar labeling, then it is possible to assign every leaf $\ell_i \in L(T(v))$ the same color; %
otherwise, the subtrees $T(u)$, $u \in \children{v}$, must receive distinct colors.

\begin{restatable}\restateref{lem:tree-consistent-planar}{lemma}{lemmaConsistentPlanar}
	\label{lem:tree-consistent-planar}
    Let $(\pi, \varphi)$ be a minimum tree-consistent colored embedding of \instance.
	For every node $v \in V$, let $\labeling_{\pi}[v]$ be the labeling $\labeling_{\pi}$ induced on the sites for leaves $\ell_i \in L(T(v))$.
	If $\labeling_{\pi}[v]$ is planar, we have $\phi(\ell_i) = \phi(\ell_j)$ for every $\ell_i, \ell_j \in L(T(v))$.
	Otherwise, $\phi(\ell_i) \neq \phi(\ell_j)$ for every $\ell_i \in L(T(u_1))$ and $\ell_j \in L(T(u_2))$, where $u_1, u_2 \in \children{v}$.
\end{restatable}
\begin{prooflater}{plemmaConsistentPlanar}
	Let $\labeling$ be a labeling of \instance and $\varphi$ an optimal tree-consistent coloring.
	Moreover, let $v \in V$ be a node of $T$.
	We consider the following two cases.
	
	\proofsubparagraph{$\boldsymbol{\labeling[v]}$ is planar.}
	Assume that there exists two leaves $\ell_i, \ell_j \in L(T(v))$ such that $\phi(\ell_i) \neq \phi(\ell_j)$.
	We now show that $\varphi$ is not optimal by creating a tree-consistent coloring $\varphi'$ with fewer colors.
	To this end, first we set $\varphi'(\ell_i) = \varphi(\ell_i)$ for every leaf $\ell_i \in L(T) \setminus L(T(v))$.
	Afterwards, we consider the embedding $\pi = \pi(\labeling[v])$.
	We set $\varphi'(\ell_i) = \varphi(\ell_i)$, where $\ell_i = \pi^{-1}(1)$ is the leftmost leaf in the embedding $\pi$ induced by $\labeling[v]$.
	Next, we consider every $i \in [n(v)]$, $i > 1$, and set $\varphi'(\pi^{-1}(i)) = \varphi'(\pi^{-1}(i - 1))$.
	
	Observe that in the end, every leaf of $T(v)$ has in $\varphi'$ the color $\varphi(\pi^{-1}(1))$.
	Since we did not alter the colors of other leaves, $\varphi'$ is tree-consistent.
	Moreover, since $\labeling[v]$ is planar, it has no monochromatic crossings.
	Finally, since now all leaves of $T(v)$ have the same color, which they did not have before, $\varphi'$ uses at least one color less than $\varphi$.
	This contradicts the assumption that $\varphi$ is optimal.
	
	\proofsubparagraph{$\boldsymbol{\labeling[v]}$ is contains a crossing.}
	Note that $\labeling[v]$ is always planar if $v$ is leaf of $T$.
	Thus, we know $v \in I$ and we let $u_1, u_2 \in \children{v}$ be the two children of $v$.
	Assume that there are twp leaves $\ell_i \in L(T(u_1))$ and $\ell_j \in L(T(u_2))$ such that $\phi(\ell_i) = \phi(\ell_j) = q$.
	Observe that $v = \lca{\ell_i}{\ell_j}$ must hold.
	Since $\varphi$ is tree-consistent, this implies that every leaf $\ell_p \in L(T(v))$ has color $\varphi(\ell_p) = c$.
	However, since $\labeling[v]$ is not planar, this implies that there exists a monochromatic crossing in $\varphi$.
	This contradicts the assumption that $\varphi$ does not have monochromatic crossings.

	The statement follows by combining the above two cases.
\end{prooflater}
\Cref{lem:tree-consistent-planar} enables us to perform a bottom-to-top DP to find a minimum tree-consistent colored embedding.
We use a DP-table $D[(v,i)]$, which stores the minimum number of colors for a tree-consistent colored embedding of the subtree $T(v)$ on the candidates $c_i$ up to $c_{i + n(v) - 1}$.
\shortLong{
Observe that the only freedom of an interior node lies in the order of its children, i.e., which child subtree's embedding starts at $c_i$; internally the subtrees are independent.
Formally, we set $D[(v, i)] = 1$ if $T(v)$ admits a planar labeling on $c_i, \ldots, c_{i + n(v) - 1}$.
Otherwise, $D[(v, i)]$ equals the minimum total colors over both orderings of \children{v}.
The table has $\BigO{n^2}$ entries and by \Cref{thm:planar} each of them can be evaluated in $\BigO{n^5}$ time.
}{%
Let $v \in I$ be an interior node of $T$ with children $u_1, u_2 \in \children{v}$.
If~$T(v)$ should be embedded on the candidates $c_i$ to $c_{i + n(v) - 1}$, then the only freedom that we have is whether $u_1$ is embedded starting at candidate $c_i$ (and $u_2$ starting at $c_{i + n(u_1)}$) or $u_2$ is embedded starting at candidate $c_i$ (and $u_1$ starting at $c_{i + n(u_2)}$).
Once this decision is made, both subtrees become independent since we seek a tree-consistent coloring.
Together with \Cref{lem:tree-consistent-planar}, this gives rise to the following recurrence.
For $D[(v, i)]$, if $T(v)$ admits a planar labeling on the candidates $c_i, \ldots, c_{i + n(v)- 1}$, we set $D[(v, i)] = 1$.
Otherwise, we set $D[(v, i)] = \min(D[(u_1, i)] + D[(u_2, i + n(u_1))], D[(u_2, i)] + D[(u_1, i + n(u_2))])$, reflecting the two embedding choices of $v$'s children.
Below, we show the running time of the DP.
}

\begin{restatable}\restateref{thm:tree-consistent}{theorem}{theoremTreeConsistent}
	\label{thm:tree-consistent}
	A minimum tree-consistent colored embedding of a geophylogenie on $n$ taxa can be determined in $\BigO{n^7}$.
\end{restatable}
\begin{prooflater}{ptheoremTreeConsistent}
    Based on \Cref{lem:tree-consistent-planar}, we use the following recurrence.
    For $D[(v, i)]$, if $T(v)$ admits a planar labeling on the candidates $c_i, \ldots, c_{i + n(v)- 1}$, we set $D[(v, i)] = 1$.
    Otherwise, we set $D[(v, i)] = \min(D[(u_1, i)] + D[(u_2, i + n(u_1))], D[(u_2, i)] + D[(u_1, i + n(u_2))])$, reflecting the two embedding choices of $v$'s children.

	We run the above-described dynamic program DP to fill the table $D$ via a bottom-to-top tree traversal of $T$.
	In the end, the minimum number of colors $k$ for a tree-consistent coloring~$\varphi$ can be read-off from the table entry $D[(v, 1)]$, where $v$ is the root of $T$.
	Correctness of this approach follows by \Cref{lem:tree-consistent-planar}.
	Regarding the running time, the table $D$ has $\BigO{n^2}$ entries.
	The time for computing a single entry $D[(v,i)]$ is dominated by the time for determining the existence of a planar embedding of $T(v)$ starting at candidate $c_i$.
	To this end, we can use for this (a slightly modified version of) our $\BigO{n^5}$-time algorithm; recall \Cref{thm:planar}.
	Therefore, the overall running time of our algorithm is $\BigO{n^7}$.	
	Note that we can obtain the embedding of $T$ in the same running time using standard backtracking techniques.
\end{prooflater}
Contrary to what this runtime bound suggests, we will see in \Cref{sec:experiments} that this algorithm easily handles instances with $n=100$ when implemented using memoization.

\section{Drawing a Single Page}
\label{sec:extensions}
So far we have considered the ``paged'' drawing of a geophylogeny as a literal coloring of the leaders, where each color class is noncrossing.
Call this the \emph{overview drawing}.
We now consider how to draw individual pages, either for browsing on a small screen (cf. Gedicke, Arzoumanidis, and Haunert~\cite{GAH.Aep.2023}), or in an interactive system that presents the overview drawing in color and lets the user ``focus'' on a particular color class.
This particularly makes sense for tree-%
\NewText{consistent} colorings, where the color classes are biologically meaningful as opposed to scattered arbitrarily throughout the tree.

Drawing a page by simply omitting the other color classes makes poor use of space: it is possible to ``zoom in'' significantly by zooming the map region to the relevant subset of sites, and by spreading the leaves across the entirety of the top boundary (presumably they are clumped in the overview drawing).
We argue that the leaf order in a page drawing should match the overview drawing, to maintain the viewer's mental model; see \Cref{fig:page-filter,fig:page-redraw}.

\begin{figure}[t]
\centering\includegraphics[page=2]{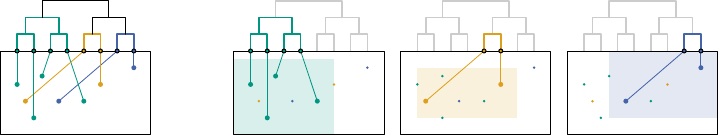}
\caption{Overview drawing with tree-consistent labeling, and its 3 page drawing instances.}
\label{fig:page-filter}
\end{figure}

\begin{figure}[t]
\centering\includegraphics[page=3]{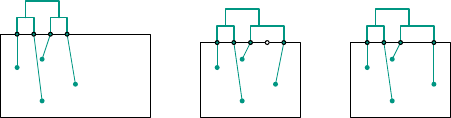}
\caption{Page drawing corresponding to the first color class of \Cref{fig:page-filter} using (from left to right): filtering, fixed candidates, and sliding labels.
Note that the middle two sliding labels are not closer together due to the  minimum-distance constraint.}
\label{fig:page-redraw}
\end{figure}

Drawing a page is therefore a distinct computational problem.
A planar labeling of a subset of $S$ is given: these sites are already associated with a reference point, implying a left-to-right total order $\prec$ on the leaves, and crossing-free leaders.
We now propose to either redistribute the leaves over the original $n$ reference points, or to use sliding labels across the top boundary.
In either case, the following observation is algorithmically useful. %

\begin{observation}
	\label{obs:page-fixed}
	Let $s_i, s_j \in S$ be two sites with $y(s_i) < y(s_j)$.
	Moreover, let $x^*$ be the $x$-coordinate of the line $\overline{s_is_j}$ through $s_i$ and $s_j$ at the top boundary of $R$.
	If $s_i \prec s_j$, then $\leader_i = (s_i, c_p)$ and $\leader_j = (s_j, c_q)$ cross if and only if $x(c_p) \geq x^*$.
	If $s_j \prec s_i$, then $\leader_i = (s_i, c_p)$ and $\leader_j = (s_j, c_q)$ cross if and only if $x(c_p) \leq x^*$.
\end{observation}

\subparagraph*{Fixed Candidates.}
We use dynamic programming that considers the candidates from left to right and stores in a DP-table $D[(i,j)]$ the objective value (e.g. leader length) of an optimal labeling on the first $i$ sites using the first $j$ candidates.
While traversing the candidates, we ensure only planar labelings are considered:
\onlyLong{Since we cannot store at which candidate each site is labeled -- this would simply be too much information to maintain in our DP table--we exploit the known label order $\prec$ of the sites once more.
}%
\Cref{obs:page-fixed} allows us to precompute fixed range of $x$-values in which the candidate for each $s_i$ must lie to ensure planarity.
\shortLong{}{This gives us a set $\text{candidates}(s_i) \subseteq C$ of \emph{feasible} candidates.
Equipped with $\text{candidates}(\cdot)$, we can now show:}

\begin{restatable}\restateref{thm:page-fixed}{theorem}{theoremPageFixed}
\label{thm:page-fixed}
    An optimal page labeling on $n$ candidate positions with respect to a leaf-additive objective function can be computed in $\BigO{n^2}$ time.
\end{restatable}
\begin{prooflater}{ptheoremPageFixed}
    Let $\instance = \instancelong$ be a geophylogenie, $v \in I$ be the root of the subtree that we want to label, $\pi \in \Pi(T(v))$ an embedding of $T(v)$ and $\prec$ the resulting left-to-right order of its leaves $L(T(v))$.
    Let the candidates $C = \{c_1, \ldots, c_n\}$ be ordered from left to right.
    Moreover, we denote with $S = \{s_1, \ldots, s_{n(v)}\}$ the sites for the leaves $L(T(v))$ and order them according to $\prec$.
    We compute in $\BigO{n^2}$-time the set of feasible candidates $\text{candidates}(s_i) \subseteq C$ for every $\ell_i \in L(T(v))$.
	Assume that we want to minimize some leaf-additive objective function $f\colon S \times C \to \mathbb{R}$.
	We now devise a DP, where the DP-table $D[(i, j)] = c$, $i \in [n(v)] \cup \{0\}$, $j \in [n] \cup \{0\}$, stores the minimum cost $c$ of a planar labeling of the first $i$ sites at the first~$j$ candidates, or $\infty$ if no planar labeling exists.
	First, note that $D[(i, 0)] = \infty$ for every $i \in [n(v)]$, since labeling sites at no candidates is impossible.
	Similarly, $D[(0, j)] = 0$ for every $j \in [n]$, since no labeling at all has cost $c$, no matter how many candidates we use.
	
	For general $i \in [n(v)]$ and $j \in [n]$, we can either label the site $s_i$ at the candidate $c_j$ (if that does not result in a crossing), or label it at some earlier candidate, i.e., one that is left of $c_j$.
	This results in the following recurrence.
	\begin{align*}
		D[(i, j)] =
		\begin{cases}
			\min (D[(i, j - 1)],\ D[(i-1, j - 1)] + f(s_i, c_j)) &\quad\text{if}\ c_j \in \text{candidates}(s_i)\\
			D[(i, j - 1)] &\quad\text{otherwise}
		\end{cases}
	\end{align*}
	Filling the table takes $\BigO{n^2}$ time and in the end $D[n(v), n]$ tells us the minimum cost of a planar labeling, or that there exists none (if $D[n(v), n] = \infty$).
	In the affirmative case we can obtain the labeling $\labeling$ in the same running time using standard backtracking techniques.
\end{prooflater}

\subparagraph*{Sliding Labels.}
For sliding labels, we allow leaves to be placed anywhere on the top boundary of the map, in the correct order $(\prec)$, at least a given distance $\Delta\geq 0$ apart, and propose to minimize the sum of squared discrepancy in $x$-coordinate between each site and the corresponding leaf.
Again by \Cref{obs:page-fixed}, we know \emph{a priori} a range of $x$ values that each particular leaf cannot go outside of.
It is easy to see that optimal label positions can be found in polynomial time by convex quadratic programming, but in fact they can be found in linear time using an isotonic regression algorithm of Kopperschmidt and Stacevi{\v{c}}ius~\cite{kopperschmidt2024pooling}

\begin{restatable}\restateref{thm:page-sliding}{theorem}{theoremPageSliding}
\label{thm:page-sliding}
    An optimal page labeling for $n$ sites on sliding candidates with least-squares $x$-discrepancy and minimum distance $\Delta$ can be computed in $\BigO{n}$ time.
\end{restatable}
\begin{prooflater}{ptheoremPageSliding}
    For every leaf $i$, let $x_i$ represent the $x$ coordinate of the leaf's position on the boundary; without loss of generality, index the leaves as $1$ to $n$ from left to right.
    To stay consistent with the given tree embedding, the sequence of $x_i$ must be monotone in $i$ (or: ``isotonic'' in the regression literature).
    To ensure a correct planar drawing, each $x_i$ is also subject to the following interval constraint.
    Firstly, $x_i$ must be on the boundary.
    Secondly, the leader must pass to the right of site $s_{i-1}$ if it exists (resp.\ to the left of $s_{i+1}$), since otherwise monotonicity implies a crossing between these leaders:
    this is a fixed interval for every $x_i$, see \Cref{obs:page-fixed}.
    
    The objective function on the $x_i$ is isotonic least squares adjustment with target values $x(s_i)$ under nonuniform interval constraints.
    This can be solved in linear time using an algorithm of Kopperschmidt and Stacevi{\v{c}}ius~\cite{kopperschmidt2024pooling}.
    To achieve a nonzero buffer distance $\Delta$ between each leaf position, let $x_i$ represent the leaf's position minus $i\cdot \Delta$ and shift its interval constraint correspondingly.
    Since the order of the leaves is given and fixed, monotonicity of the $x_i$ now encodes the buffered monotonicity of the leaf positions.
\end{prooflater}

\section{Integer Linear Programming Formulations}
\label{sec:ilp}
Given the conjectured hardness of finding a minimum colored embedding for arbitrary geophylogenies~\cite{DKMW.GLG.2023,Ung.kCC.1988}, we discuss here a number of integer linear programs (ILP)
based on Klawitter et al.'s~\cite{KKS+.VGI.2025} ILP for obtaining a crossing-minimal labeling of (\ILPCrossing).  %
\ILPCrossing %
has a binary variable $\chi_{i,j}$ for every pair of sites $s_i, s_j \in S$, $i < j$, which indicates whether the two leaders $\leader_i$ and $\leader_j$ are \emph{allowed to} cross. 
We %
connect our coloring logic to the $\chi$-variables by requiring different colors whenever $\chi_{i,j} = 1$.

\subparagraph*{Minimizing Colors (\ILPColors).}
We extend \ILPCrossing with %
a binary variable $\varphi_{i,c}$ for every leaf $\ell_i$ and color $c \in [n]$ to indicate if leaf $\ell_i$ receives color $c$, i.e., $\varphi(\ell_i) = c$.
Note that $n$ is a trivial upper bound for the number of required colors.
Moreover, we use a binary variable %
$\gamma_{c}$ for every color $c \in [n]$ to indicate some leaf has color $c$.

We minimize $\sum_{c \in [n]} \gamma_c$ and add the following constraints to \ILPCrossing. %

\begin{align}
    \sum_{c \in [n]} \varphi_{i,c} = 1 &\qquad\text{for all}\ \ell_i \in L(T) \label{con:ilp-colors-one-color}\\
    \varphi_{i,c} + \varphi_{j,c} + \chi_{i,j} \leq 2 &\qquad\text{for all}\ \ell_i, \ell_j \in L(T),\ i < j,\ c \in [n] \label{con:ilp-colors-crossing}\\
    \varphi_{i,c} \leq \gamma_c &\qquad\text{for all}\ \ell_i \in L(T),\ c \in [n] \label{con:ilp-color-used}
\end{align}
Constraint~\eqref{con:ilp-colors-one-color} ensures that every leaf receives exactly one color, Constraint~\ref{con:ilp-color-used} forces $\gamma_c = 1$ whenever $\varphi_{i,c} = 1$ for some $\ell_i \in L(T)$. 
Constraint~\eqref{con:ilp-colors-crossing} ensures that two crossing leaders $\leader_i$ and $\leader_j$ receive different colors since \ILPCrossing guarantees $\chi_{i,j} = 1$ in this case (for $\chi_{i,j} = 0$, this constraint is trivially satisfied).
Recall that $\chi_{i,j} = 1$ is possible for non-crossing leaders and $\gamma_c = 1$ can hold even if no leaf has color $c \in [n]$.
This is not problematic for correctness since we minimize  $\sum_{c \in [n]} \gamma_c$, and
\ILPCrossing guarantees for $\chi_{i,j} = 0$ non-crossing leaders. %

\subparagraph*{Other Formulations.}
We additionally implement a formulation \ILPColorsSym which adds known symmetry breaking constraints for graph coloring~\cite{MDZ.PAG.2001}.
The literature contains further formulations for coloring that are tailored to graphs of certain density and usually outperform the above textbook formulation~\cite{JM.NIL.2018}.
Since our leader intersection graph (and its density) varies with the embedding of $T$, it is a priori not clear which formulation will perform better in our setting.
Therefore, we implement two further alternative formulations:
the \emph{representative formulation}~\cite{CCC.arf.2008}, where each leaf $\ell_i$ can become the representative of its color class, and the \emph{partial order based formulation}~\cite{JM.NIL.2018}, in which we fix a linear order on the color classes.
Call these \ILPColorsRep and \ILPColorsPop respectively. %

\subparagraph*{Tree-Consistent Colorings.}
In our experimental evaluation, we will also evaluate the number of crossings in a tree-consistent coloring with only a limited number of \emph{available} colors.
Thus, we now discuss how to incorporate tree-consistency into formulation \ILPColors. %
We introduce a binary variable $\phi_{v,c} \in \{0,1\}$ for every internal node $v \in I$.
If $\phi_{v,c} = 1$, then every leaf $\ell_i \in L(T(v))$ of the subtree $T(v)$ should receive color~$c$.
Furthermore, we add the constraints:
\begin{align}
    \sum_{c \in [n]} \varphi_{v,c} \leq 1 &\qquad\text{for every}\ v \in I \label{con:ilp-colors-tree-one-color}\\
    \varphi_{v,c} \leq \varphi_{u,c} &\qquad\text{for all}\ v \in I,\ u \in \children{v},\ c \in [n] \label{con:ilp-colors-tree-children}\\
    \varphi_{i,c} + \varphi_{j,c} - 1 \leq \varphi_{v,c} &\qquad\text{for all}\ \ell_i, \ell_j \in L(T),\ v=\lca{\ell_i}{\ell_j},\ i < j,\ c \in [n] \label{con:ilp-colors-tree-lca}
\end{align}
Note that Constraint~\eqref{con:ilp-colors-tree-one-color} allows $\phi_{v,c} = 0$ for all colors $c\in [n]$, which models the case that leaves of $L(T(v))$ receive different colors.
Constraints~\eqref{con:ilp-colors-tree-children} and \eqref{con:ilp-colors-tree-lca} propagate colors down to children and up to the lowest common ancestor, respectively, to ensure tree-consistency.

\section{Experimental Evaluation}
\label{sec:experiments}
In this section, we experimentally evaluate the performance of our algorithms in terms of their running time and compare the effect of insisting on tree-consistency on the number of required colors.
We implemented the dynamic program (DP) for computing optimal tree-consistent colorings from \Cref{sec:tree-consistent} %
and all ILP formulations from \Cref{sec:ilp} in \texttt{Python 3.12}.\footnote{The source code and the files for the experimental evaluation are available on \href{https://osf.io/rbqd3/overview?view_only=09e34b73481c4fbb973b8158273cbc00}{OSF}~\cite{SUPPLEMENTARY_MATERIAL}.
The application can also be accessed online at \url{https://www.ac.tuwien.ac.at/projects/paged-geophylogenies/}.}
We used \texttt{Gurobi 13}~\cite{Gurobi} as ILP solver and based our ILP-based solutions on the publicly available implementation of \ILPCrossing by Klawitter et al.~\cite{GeoPhylo}.

\subparagraph*{Dataset.}
We use the available \emph{synthetic} instances of Klawitter et al.~\cite{KKS+.VGI.2025}, consisting of $570$ geophylogenies with up to one hundred sites.
These are three sets of randomly generated instances that differ in the way the sites are placed on the map.
The \emph{uniform} instances place them uniformly at random, for \emph{coastline} instances they are initially placed along a horizontal line at uniform distance and then perturbed in $x$- and $y$-direction to simulate a distribution along a coastline or river bank, and \emph{clustered} instances contain clusters of size between three and ten, placed uniformly at random on the map, where the sites for each cluster are placed inside a disk around its center.
For all instance types, the phylogenetic tree $T$ is constructed by iteratively merging sites or subtrees into a larger tree with a probability that is inverse to the distance between the sites or median position of the subtree.
For clustered instances, the trees are first created for each cluster separately before they are merged into the final tree $T$.

For space reasons, we focus here on the clustered and coastline instances, since they have a real-world motivation.
We provide the plots for uniform instances in \Cref{app:experiments}.

\subparagraph*{Setup and Implementation Details.}
All algorithms ran on a system with Intel Xeon E5-2640 v4 10-core processors at 2.40 GHz; the ILP solver was allowed to use all available cores, but our Python code is single-threaded. 
There was a hard memory limit of 96GB, but this was never an issue.
We set a solver timeout of one hour unless specified otherwise.
When reporting running times, we report the solve time, i.e., ignoring the time to build the model.
Although there was a noticeable build time for larger instances, it was negligible compared to the solve time. 
For the DP, we report wall-clock running times, including the time to precompute the lookup tables but ignoring the time required to read the instance.
We do not expect to query lowest common ancestor of every leaf pair, and, therefore, opted for a na{\"i}ve computation.
The DP from \Cref{sec:planar} is implemented using memoization (``top down'') rather than filling a table completely, since we expect many states to be irrelevant. %

\subsection{Comparison of the ILP Variants}

\subparagraph*{Integer Linear Programming is Fast Enough for Small Instances.}
\Cref{fig:formulations-running-times} shows the running time of each formulation on each instance.
For small instances ($n \leq 30$), all variants finish within ten seconds.
Starting at about $n=40$, the solver times out on some instances.

All formulations are slower than \ILPCrossing, which terminated within 100 seconds even on many of the largest instances~\cite{KKS+.VGI.2025}.
This is unsurprising, since we are extending an \NP-hard crossing minimization problem with a coloring problem.
Overall, clustered instances seem to be easier than coastline instances, possibly due to the favorable phylogenetic tree structure of clustered instances, which ensures that each cluster appears as a subtree of $T$.

As for the variants of the ILP, a Wilcoxon signed-rank test on these runtime data shows that \ILPColors and \ILPColorsRep perform significantly worse than \ILPColorsPop and \ILPColorsSym ($p\approx 0$), but within these pairs the difference is clearly not significant ($p>25\%$).
This grouping is particularly obvious in the Kaplan-Meier chart of clustered instances (\Cref{fig:kaplanmeier}).

\begin{figure}[t]
	\centering
	\includegraphics[page=1,width=\linewidth]{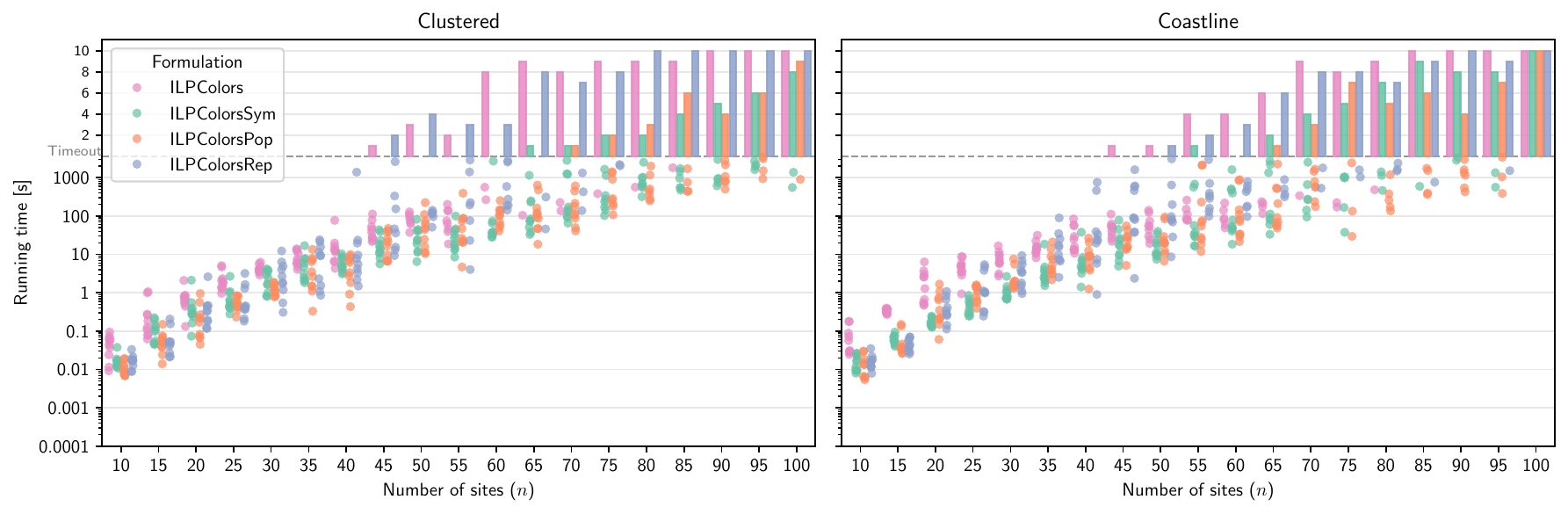}
	\caption{Solve times ($\log_{10}$-scale) of the different ILP formulations that terminated within one hour. We display above the plot the number of timed out instances for each formulation.}
	\label{fig:formulations-running-times}
\end{figure}

\begin{figure}[b]
	\centering
	\includegraphics[width=\linewidth]{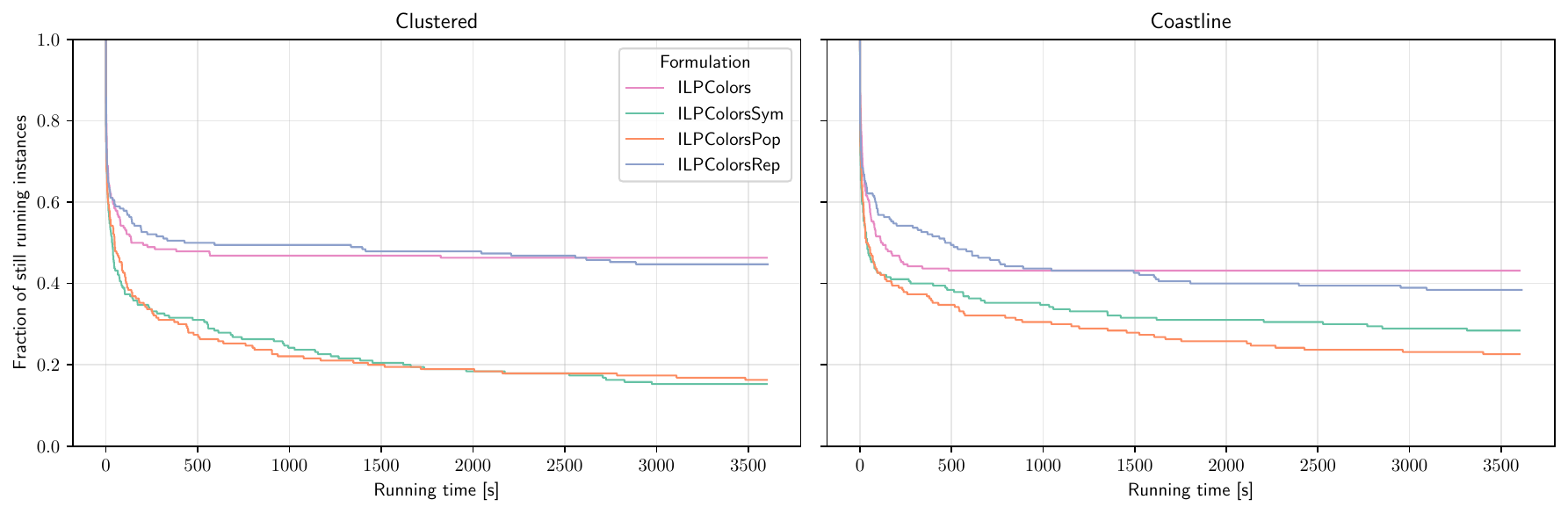}
	\caption{Survival plot of clustered and coastline instances; recall the timeout of 3600 seconds.}
    \label{fig:kaplanmeier}
\end{figure}

\subsection{The Impact of Tree-Consistent Colorings}
Next, we consider tree-consistent colorings and evaluate the running time of our DP as well as the number of additional colors needed to ensure tree-consistency.

\subparagraph*{The DP is Very Fast in Practice.}
\Cref{fig:treedp-running-times} shows the running times of dynamic program from \Cref{sec:tree-consistent} for computing the minimum number of colors required in a tree-consistent colorings.
We can see a stark contrast to the running times seen in \Cref{fig:formulations-running-times}: it runs within one second for nearly all instances.
The overall trend that clustered instances admit faster running times can be observed here as well---if less prominently. %
Moreover, \Cref{fig:treedp-running-times} indicates that the practical running time of our DP is far better than what the theoretical bound of $\BigO{n^7}$ might suggest. %
In particular, with running times in a fraction of a second, it is suited for interactive applications.
Recall that the algorithm behind \Cref{thm:tree-consistent} has to determine the planarity of many subtrees of $T$ on varying candidate ranges.
We have, by subsumption, thus also demonstrated that our \BigO{n^5}-time planarity test is fast in practice.
Although we are not aware of any implementations of the Klawitter et al.~$\BigO{n^6}$ planarity test~\cite{KKS+.VGI.2025},  we do not expect it to outperform ours.
The fast running time suggests that many of the states that exist in theory are never evaluated in practice.
Thus, we plot in \Cref{fig:treedp-states} for each performed planarity test the overall number of evaluated states.
We can observe that we only consider a tiny fraction of the (theoretical) $\BigO{n^4}$-many states.
Furthermore, we can notice a few ``prominent'' linear functions, which indicates that (in)feasibility was often detected after labeling some $k$ bottommost sites.

\begin{figure}[t]
	\centering
	\includegraphics[page=1,width=\linewidth]{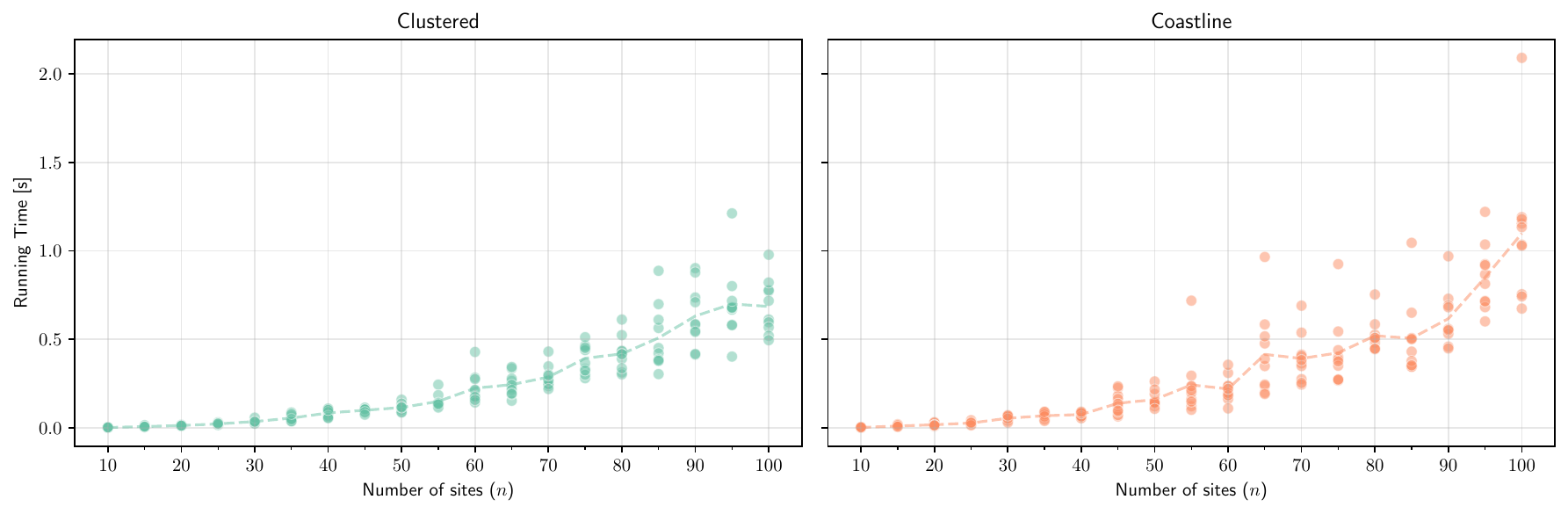}
	\caption{Wall-clock running times of our dynamic program from \Cref{sec:tree-consistent}. The dashed line depicts the average running time.}
	\label{fig:treedp-running-times}
\end{figure}

\begin{figure}[b]
	\centering
	\includegraphics[page=1,width=\linewidth]{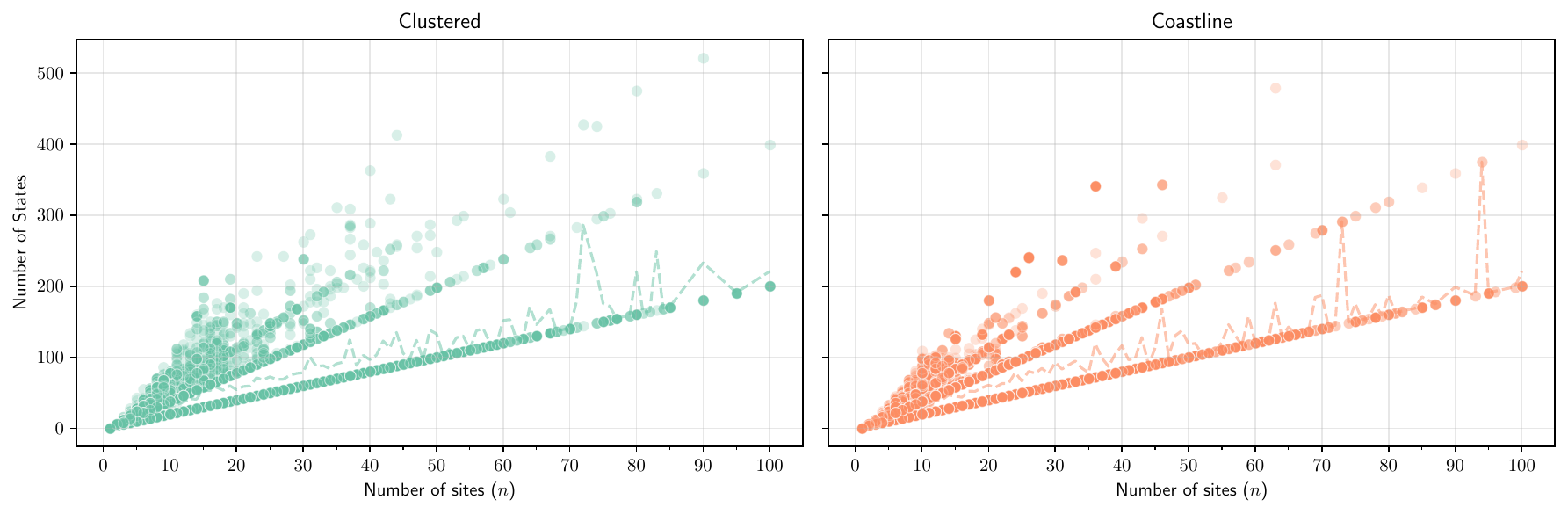}
	\caption{Number of evaluated states in the individual planarity tests performed by our dynamic program from \Cref{sec:tree-consistent}. The dashed line depicts the average. Note that instance-sizes that are not multiples of five arise due to considering also sub-instances.}
	\label{fig:treedp-states}
\end{figure}

\subparagraph*{The Price of Tree-Consistency.}
\onlyLong{Being tree-consistent is an additional requirement that, in general, increases the number of colors required to avoid monochromatic crossings.
However, it is unclear, and highly instance-dependent, to what extent this increases is present.}
In \Cref{fig:tree-dp-colors}, we plot the minimum number of colors required for a tree-consistent coloring for each of the instances.
Somewhat surprisingly, the number of required colors seems to increase linearly with the number of sites in the instance; see also the dashed line that shows the average number of colors for each $n$.
Moreover, coastline instance seem to require more colors in general, which indicates that their almost-horizontal placement of vertices is an obstruction to planarity.
To quantify this ``price'' of tree-consistency, we plot in \Cref{fig:tree-vs-normal-colors} the relative number of tree-consistent colors compared to the unconstrained coloring for instances with $n \leq 50$ (since for larger instances we do not necessarily have the true minimum colored embedding).
For our smallest instance, the price of tree-consistency is manageable, and we need on average less than twice as many colors.
For larger instances, not only the relative number of additional colors increases, but also the variance among the distances, which indicates a dependency of the site-placement on the planarity of sub-trees and thus the number of required colors in a tree-consistent coloring. %

\begin{figure}[t]
	\centering
	\includegraphics[page=1,width=\linewidth]{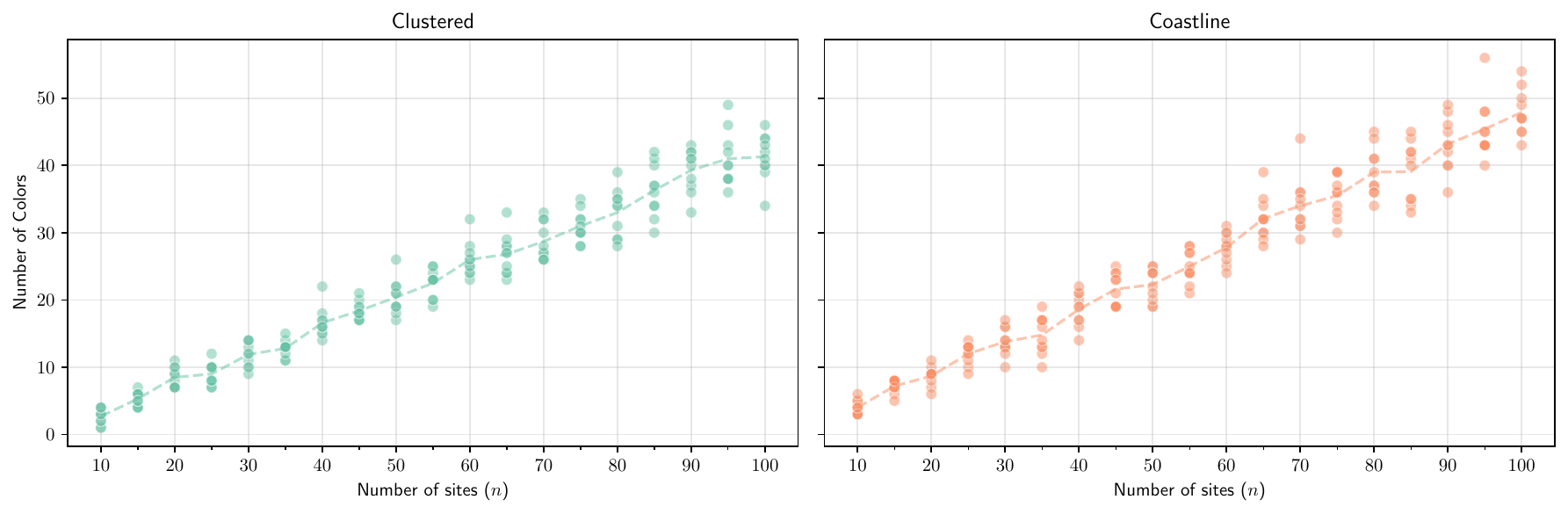}
	\caption{The number of colors of an optimal tree-consistent coloring.}
	\label{fig:tree-dp-colors}
\end{figure}

\begin{figure}[t]
	\centering
	\includegraphics[page=1,width=\linewidth]{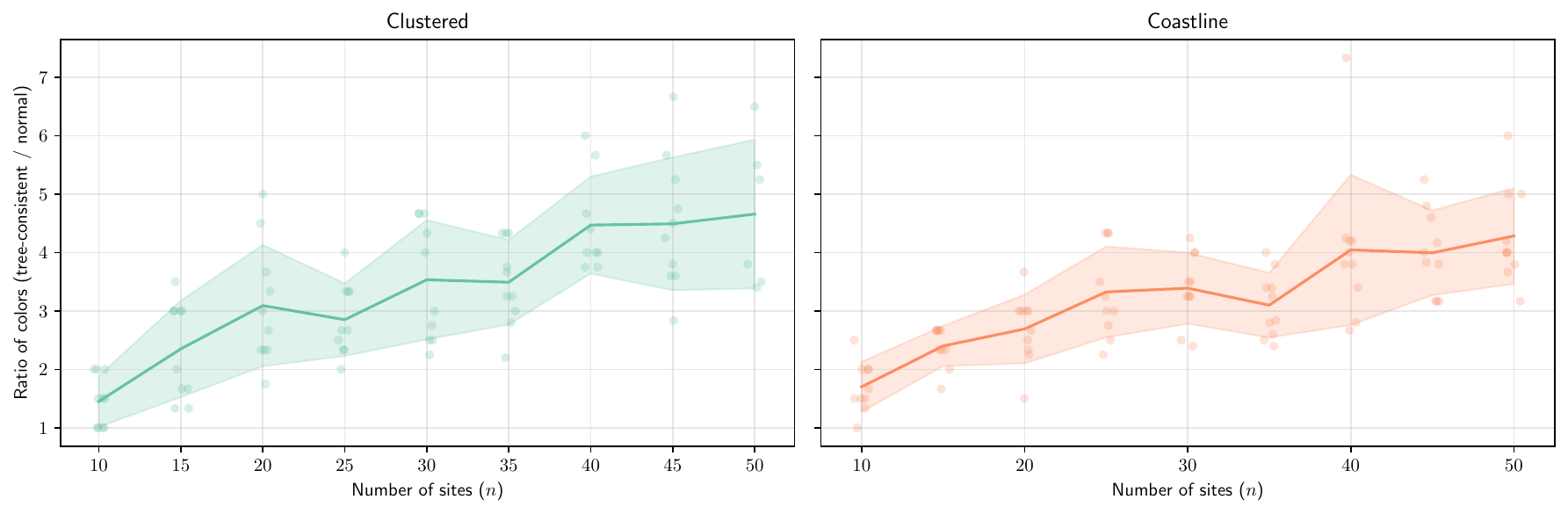}
	\caption{Relative number of colors required for an optimal tree-consistent coloring compared to an unconstrained coloring. Dots represent individual instances and the line shows the average relative number together with the standard deviation.}
	\label{fig:tree-vs-normal-colors}
\end{figure}

\subsection{The Benefit of Additional Colors}
Until now we considered two extreme cases: No colors but many (monochromatic) crossings and many colors but no monochromatic crossings.
With this final experiment, we investigate the tradeoff between these extremes and analyze the scenario where we only have a limited \emph{budget} of pages (colors), and have to accept some monochromatic crossings.

\subparagraph*{Extending the ILP.}
Recall that our ILP from \Cref{sec:ilp} is targeted at minimizing the number of colors required to avoid \emph{all} monochromatic crossings.
We now briefly mention how to extend it to minimize the number of monochromatic crossings for a limited budget of colors.
To this end, we do not introduce for every leaf $\ell_i$ $n$ color variables $\varphi_{i,c}$, $c \in [n]$, but only budget $b$-many.
Moreover, we introduce for every site-pair $s_i, s_j \in S$, $i < j$, and every color $c \in [b]$ a variable $\chi_{i,j}^c$, which is one if there is a monochromatic crossing between the leaders $\leader_i$ and $\leader_j$ in color $c$.
Finally, we add the constraint 
\begin{align}
	\chi_{i,j} + \varphi_{i,c} + \varphi_{j,c} - 2 \leq \chi_{i,j}^c &\qquad\text{for all}\ \ell_i < \ell_j \in L(T),\ i < j,\ c \in [b]\label{con:ilp-minimize-monochromatic-crossings},
\end{align}
which links the new variables to the old one.
Observe that Constraint~\ref{con:ilp-minimize-monochromatic-crossings} only enforces $\chi_{i,j}^c = 1$ if $\chi_{i,j} = \varphi_{i,c} = \varphi_{j,c} = 1$.
However, we minimize the sum over all $\chi_{i,j}^c$-variables.

\subparagraph*{Even a Few Colors Help a Lot.}
\begin{figure}[t]
	\centering
	\includegraphics[page=1,width=\linewidth]{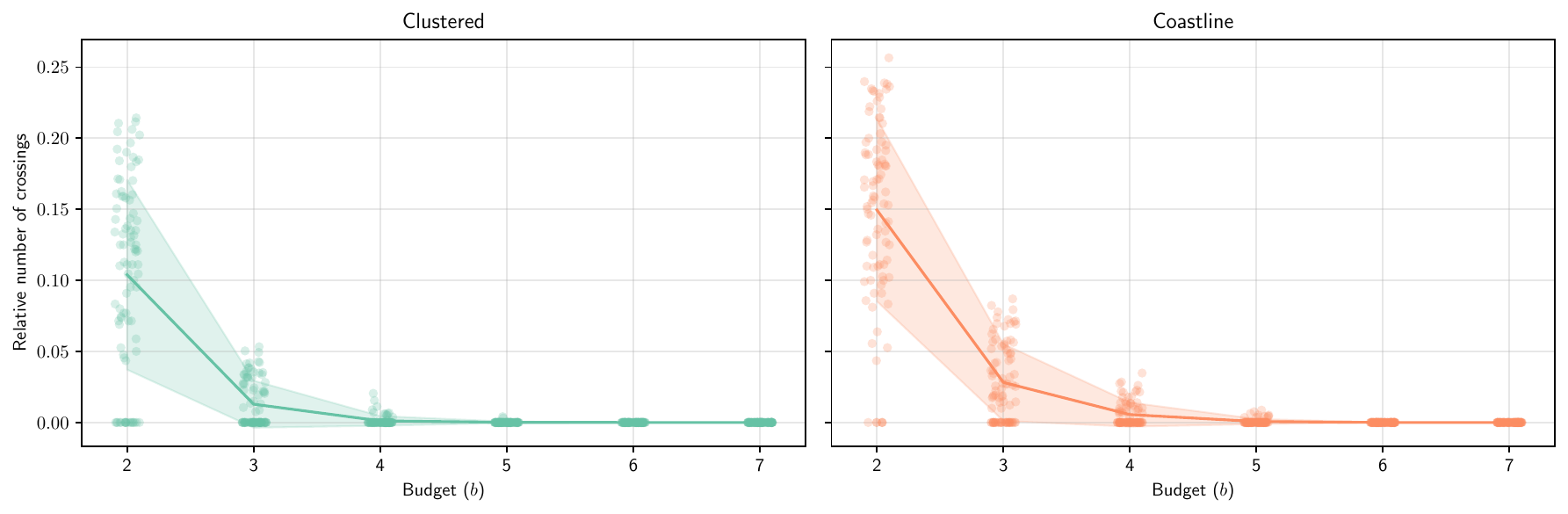}
	\caption{Relative number of monochromatic crossings compared to the crossing-minimal labeling for different budgets $b$. Dots represent individual instances and the line shows the average relative number together with the standard deviation.}
	\label{fig:tradeoff-all-colors}
\end{figure}

\begin{figure}[t]
	\centering
	\includegraphics[page=1,width=\linewidth]{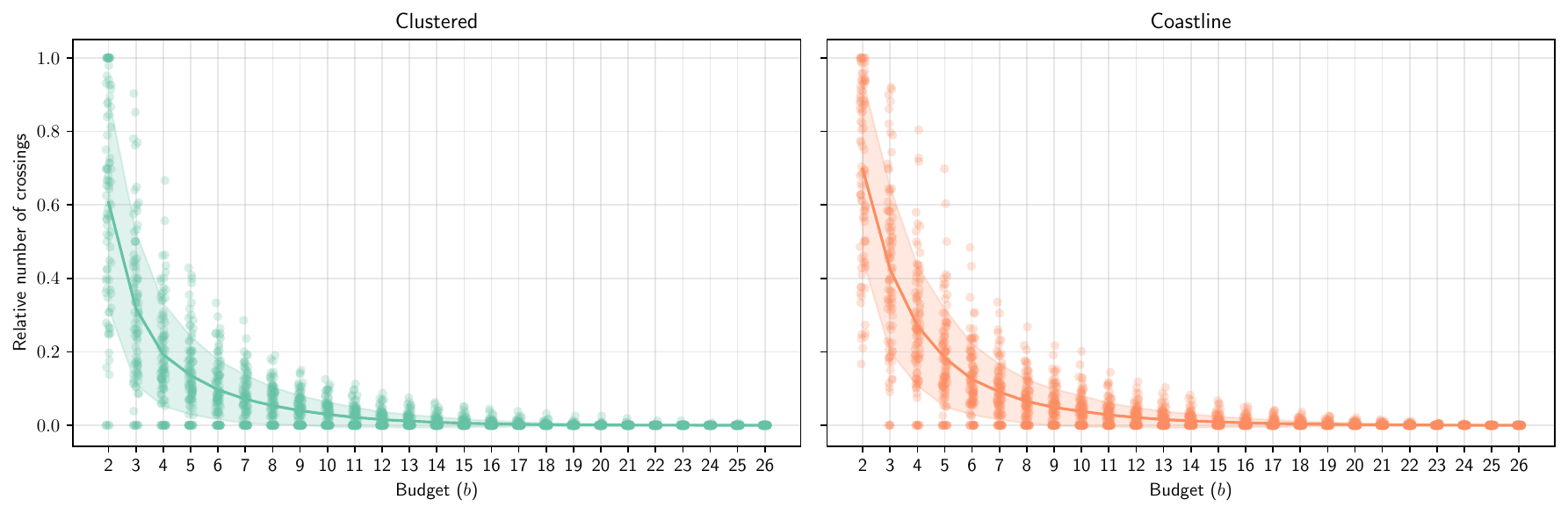}
	\caption{Relative number of crossings for tree-consistent colorings compared for different budgets~$b$. Symbols have the same meaning as in \Cref{fig:tradeoff-all-colors}.}
	\label{fig:tradeoff-tree-colors}
\end{figure}

We ran \ILPColorsSym with the above-described extension on every instance with $n \leq 50$ and for every color budget $b$, starting with $b = 1$, i.e., only one color, and we just minimize the overall number of crossing, and ending where the number of monochromatic crossings becomes zero.
Similarly, we re-ran \ILPColorsSym with the above-described extension but this time for tree-consistent colorings; recall \Cref{sec:ilp} where we describe how to extend \ILPColorsSym in this regard.
\Cref{fig:tradeoff-all-colors} shows, for general colorings, the relative number of monochromatic crossings for different budgets $b$ compared to their number for $b = 1$.
We can observe that already adding a single color brings us down to $10\%$ of the original monochromatic crossings on average.
With three and four colors, there are only relatively few initial monochromatic crossing left and with only seven colors, we could avoid all monochromatic crossings in every instance; for clustered and coastline instances this was already the case for six colors.

In \Cref{fig:tradeoff-tree-colors}, we perform the same comparison for tree-consistent colorings.
There, we can see that we need more colors to avoid all monochromatic crossings (up to 26).
However, also here three colors are sufficient to have only forty percent of the original monochromatic crossings.
With seven colors we only have around ten percent of the original colors and adding just three more colors, nearly all monochromatic crossings are gone.

Overall, we can see a clear tradeoff between minimizing the number of colors and the number of (monochromatic) crossings.
Somewhat surprisingly, it turns out that with just a few extra pages, we can substantially reduce the number of crossings on each page, and we believe that this should also improve the overall readability of the labeling.
As an example, \Cref{fig:tradeoff-fish} shows a version of the real-world instance of \Cref{fig:intro} using only two colors.
\NewText{
More generally, our experiments show that we can eliminate almost all monochromatic crossings by using just ten colors, which is of particular relevance for the overview drawings that we have discussed in \Cref{sec:extensions}.
Recall that research suggests using at most ten colors to encode nominal data reliably, in particular when used against varying backgrounds, such as on cartographic maps~\cite{War.CFC.2021}.
}

\begin{figure}[t]
	\centering
	\includegraphics[page=1]{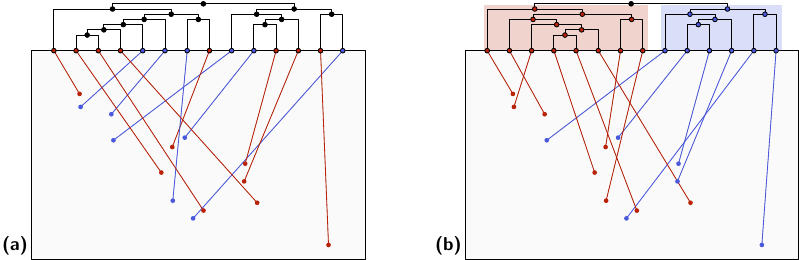}
	\caption{Colored embedding of the instance of \Cref{fig:intro} with the fewest monochromatic crossings for \textbf{\textsf{(a)}} unrestricted (2 crossings) and \textbf{\textsf{(b)}} tree-consistent colorings (7 crossings) with two colors.}
	\label{fig:tradeoff-fish}
\end{figure}

\section{Discussion and Concluding Remarks}
\label{sec:conclusion}
In this paper, we introduced the page metaphor for geophylogeny drawings: we color the leaders to avoid monochromatic crossings, i.e., each color class induces a page which features a planar sublabeling.
Our experimental evaluation shows that integer linear programming can be used to find an optimal solution in a few seconds for small instances, but can take minutes up to hours for instances with more than $50$ taxa.
If each color class is required to consist of a single subtree of the geophylogeny---and from a visualization perspective we argue it should---then optimal solutions can be computed in seconds even for the largest instances by dynamic programming.
Realistically, instances in boundary labeling have fewer than 50 sites~\cite{DNTW.Cbl.2025,NNR.RCL.2017}; even crossing-minimal labelings of large instances tend to be unreadable~\cite{KKS+.VGI.2025} and we expect that significantly larger instances remain unreadable even for colored leaders, hence our methods solve all reasonable instances efficiently.

Our experimental study also revealed that tree-consistent colorings require significantly more colors than unrestricted ones.
Therefore, we investigated the tradeoff between additional colors and the number of monochromatic crossings.
We show that by introducing at most five pages, only 20\% of the original (monochromatic) crossings remain in the tree-consistent case, and much fewer still in the unrestricted case.
This tradeoff, as well as the effectiveness of our coloring approach versus a paged approach, should be examined in a user study.
\NewText{In particular, we see conducting an expert evaluation with biologists which also evaluates the (dis)advantages of overview drawings and drawings of a single page as a natural next step.}

On the theory side, open questions include the complexity of the unrestricted colored embedding problem, %
and a better analysis of the worst-case runtime of the planarity DP.
Finally, this paper considers only straight-line leaders.
While we suspect that our results carry over for L-shaped ``po'' leaders, extending our work to other leader styles or radial maps is an interesting direction for future work.

\bibliography{references}

@InProceedings{Ung.kCC.1988,
  author    = {Walter Unger},
  booktitle = {Proc. 5th Symposium on Theoretical Aspects of Computer Science (STACS'88)},
  title     = {On the {$k$}-{C}olouring of {C}ircle-{G}raphs},
  year      = {1988},
  editor    = {Robert Cori and Martin Wirsing},
  pages     = {61--72},
  publisher = {Springer},
  series    = {Lecture Notes in Computer Science},
  volume    = {294},
  doi       = {10.1007/BFB0035832},
}

@Article{DKMW.GLG.2023,
  author  = {James Davies and Tomasz Krawczyk and Rose McCarty and Bartosz Walczak},
  journal = {Discrete \& Computational Geometry},
  title   = {Grounded L-Graphs Are Polynomially {\(\chi\)}-Bounded},
  year    = {2023},
  number  = {4},
  pages   = {1523--1550},
  volume  = {70},
  doi     = {10.1007/s00454-023-00592-z},
}

@article{kopperschmidt2024pooling,
  title={Pooling adjacent violators under interval constraints},
  author={Kopperschmidt, Kai and Stacevi{\v{c}}ius, Ruslanas},
  journal={Optimization Letters},
  volume={18},
  number={1},
  pages={257--277},
  year={2024},
  publisher={Springer}
}

@article{fish,
author = {Williams, Trevor J. and Johnson, Jerald B.},
title = {History predicts contemporary community diversity within a biogeographic province of freshwater fish},
journal = {Journal of Biogeography},
volume = {49},
number = {5},
pages = {809-821},
doi = {https://doi.org/10.1111/jbi.14316},
url = {https://onlinelibrary.wiley.com/doi/abs/10.1111/jbi.14316},
eprint = {https://onlinelibrary.wiley.com/doi/pdf/10.1111/jbi.14316},
year = {2022}
}

@InProceedings{NNR.RCL.2017,
  author    = {Niedermann, Benjamin and N{\"o}llenburg, Martin and Rutter, Ignaz},
  booktitle = {Proc. 10th IEEE Pacific Visualization Symposium (PacificVis)},
  title     = {{R}adial {C}ontour {L}abeling with {S}traight {L}eaders},
  year      = {2017},
  pages     = {295--304},
  publisher = {{IEEE} Computer Society},
  series    = {PacificVis '17},
  doi       = {10.1109/PACIFICVIS.2017.8031608},
}

@Book{BNN.ELF.2021,
  author    = {Michael A. Bekos and Benjamin Niedermann and Martin N{\"{o}}llenburg},
  publisher = {Springer},
  title     = {{E}xternal {L}abeling: {F}undamental {C}oncepts and {A}lgorithmic {T}echniques},
  year      = {2021},
  series    = {Synthesis Lectures on Visualization},
  doi       = {10.1007/978-3-031-02609-6},
}

@Article{AvKS.Lpm.1998,
  author  = {Agarwal, Pankaj K. and {van Kreveld}, Marc J. and Suri, Subhash},
  journal = {{Computational Geometry}},
  title   = {{Label placement by maximum independent set in rectangles}},
  year    = {1998},
  issn    = {09257721},
  number  = {3-4},
  pages   = {209--218},
  volume  = {11},
  doi     = {10.1016/S0925-7721(98)00028-5},
}

@Article{vKSW.Pls.1999,
  author  = {{van Kreveld}, Marc J. and Strijk, Tycho and Wolff, Alexander},
  journal = {{Computational Geometry}},
  title   = {{Point labeling with sliding labels}},
  year    = {1999},
  issn    = {09257721},
  number  = {1},
  pages   = {21--47},
  volume  = {13},
  doi     = {10.1016/S0925-7721(99)00005-X},
}

@Book{Bri.RGS.1990,
  author    = {Mary Helen Briscoe},
  publisher = {Springer Science \& Business Media},
  title     = {{A} {R}esearcher's {G}uide to {S}cientific and {M}edical {I}llustrations},
  year      = {1990},
  doi       = {10.1007/978-1-4684-0355-8},
}

@Article{BKNS.BLO.2010,
  author  = {Michael A. Bekos and Michael Kaufmann and Martin N{\"{o}}llenburg and Antonios Symvonis},
  journal = {Algorithmica},
  title   = {{B}oundary {L}abeling with {O}ctilinear {L}eaders},
  year    = {2010},
  number  = {3},
  pages   = {436--461},
  volume  = {57},
  doi     = {10.1007/s00453-009-9283-6},
}

@Article{BHKN.AMC.2009,
  author  = {Marc Benkert and Herman J. Haverkort and Moritz Kroll and Martin N{\"{o}}llenburg},
  journal = {Journal of Graph Algorithms and Applications (JGAA)},
  title   = {{A}lgorithms for {M}ulti-{C}riteria {B}oundary {L}abeling},
  year    = {2009},
  number  = {3},
  pages   = {289--317},
  volume  = {13},
  doi     = {10.7155/jgaa.00189},
}

@Article{BKSW.BlM.2007,
  author  = {Michael A. Bekos and Michael Kaufmann and Antonios Symvonis and Alexander Wolff},
  journal = {Computational Geometry},
  title   = {{B}oundary labeling: {M}odels and efficient algorithms for rectangular maps},
  year    = {2007},
  number  = {3},
  pages   = {215--236},
  volume  = {36},
  doi     = {10.1016/j.comgeo.2006.05.003},
}

@InProceedings{KLW.LBL.2014,
  author    = {Philipp Kindermann and Fabian Lipp and Alexander Wolff},
  booktitle = {Proc. 22nd International Symposium Graph Drawing and Network Visualization (GD)},
  title     = {{L}uatodonotes: {B}oundary {L}abeling for {A}nnotations in {T}exts},
  year      = {2014},
  pages     = {76--88},
  publisher = {Springer},
  series    = {Lecture Notes in Computer Science (LNCS)},
  volume    = {8871},
  doi       = {10.1007/978-3-662-45803-7_7},
}

@InProceedings{HPL.BLF.2014,
  author    = {Zhi{-}Dong Huang and Sheung{-}Hung Poon and Chun{-}Cheng Lin},
  booktitle = {Proc. 8th International Conference and Workshops on Algorithms and Computation (WALCOM)},
  title     = {{B}oundary {L}abeling with {F}lexible {L}abel {P}ositions},
  year      = {2014},
  pages     = {44--55},
  publisher = {Springer},
  series    = {Lecture Notes in Computer Science (LNCS)},
  volume    = {8344},
  doi       = {10.1007/978-3-319-04657-0_7},
}

@Article{BGNN.rlb.2019,
  author  = {Lukas Barth and Andreas Gemsa and Benjamin Niedermann and Martin N{\"{o}}llenburg},
  journal = {Information Visualization},
  title   = {On the readability of leaders in boundary labeling},
  year    = {2019},
  number  = {1},
  pages   = {110--132},
  volume  = {18},
  doi     = {10.1177/1473871618799500},
}

@Article{GAH.Aep.2023,
  author  = {Sven Gedicke and Lukas Arzoumanidis and Jan-Henrik Haunert},
  journal = {Cartography and Geographic Information Science (CaGIS)},
  title   = {Automating the external placement of symbols for point features in situation maps for emergency response},
  year    = {2023},
  number  = {4},
  pages   = {385--402},
  volume  = {50},
  doi     = {10.1080/15230406.2023.2213446},
}

@InProceedings{KKS+.VGI.2023,
  author    = {Jonathan Klawitter and Felix Klesen and Joris Y. Scholl and Thomas C. van Dijk and Alexander Zaft},
  booktitle = {Proc. 12th International Conference Geographic Information Science (GIScience)},
  title     = {{V}isualizing {G}eophylogenies - {I}nternal and {E}xternal {L}abeling with {P}hylogenetic {T}ree {C}onstraints},
  year      = {2023},
  pages     = {5:1--5:16},
  publisher = {Schloss Dagstuhl - Leibniz-Zentrum f{\"{u}}r Informatik},
  series    = {Leibniz International Proceedings in Informatics (LIPIcs)},
  volume    = {277},
  doi       = {10.4230/LIPIcs.GIScience.2023.5},
}

@InProceedings{Pur.WAh.1997,
  author    = {Helen C. Purchase},
  booktitle = {Proc. 5th International Symposium on Graph Drawing and Network Visualization (GD'97)},
  title     = {Which Aesthetic has the Greatest Effect on Human Understanding?},
  year      = {1997},
  editor    = {Giuseppe Di Battista},
  pages     = {248--261},
  publisher = {Springer},
  series    = {Lecture Notes in Computer Science},
  volume    = {1353},
  doi       = {10.1007/3-540-63938-1_67},
}

@InProceedings{BNT+.BLC.2024,
  author    = {Bonerath, Annika and Nöllenburg, Martin and Terziadis, Soeren and Wallinger, Markus and Wulms, Jules},
  booktitle = {Proc. 32nd International Symposium on Graph Drawing and Network Visualization (GD'24)},
  title     = {Boundary Labeling in a Circular Orbit},
  year      = {2024},
  editor    = {Stefan Felsner and Karsten Klein},
  pages     = {22:1--22:17},
  publisher = {Schloss Dagstuhl – Leibniz-Zentrum für Informatik},
  series    = {LIPIcs},
  volume    = {320},
  doi       = {10.4230/LIPICS.GD.2024.22},
}

@Article{DNTW.Cbl.2025,
  author  = {Depian, Thomas and Nöllenburg, Martin and Terziadis, Soeren and Wallinger, Markus},
  journal = {Computational Geometry},
  title   = {Constrained boundary labeling},
  year    = {2025},
  pages   = {102191},
  volume  = {129},
  doi     = {10.1016/j.comgeo.2025.102191},
}

@InProceedings{DNTW.CBL.2024,
  author    = {Depian, Thomas and Nöllenburg, Martin and Terziadis, Soeren and Wallinger, Markus},
  booktitle = {Proc. 35th International Symposium on Algorithms and Computation (ISAAC'24)},
  title     = {Constrained Boundary Labeling},
  year      = {2024},
  editor    = {Juli{\'{a}}n Mestre and Anthony Wirth},
  pages     = {26:1--26:16},
  publisher = {Schloss Dagstuhl – Leibniz-Zentrum für Informatik},
  series    = {LIPIcs},
  volume    = {322},
  doi       = {10.4230/LIPICS.ISAAC.2024.26},
}

@Article{KKS+.VGI.2025,
  author  = {Klawitter, Jonathan and Klesen, Felix and Scholl, Joris Y. and Van Dijk, Thomas C. and Zaft, Alexander},
  journal = {Journal of Graph Algorithms and Applications},
  title   = {Visualizing Geophylogenies - Internal and External Labeling with Phylogenetic Tree Constraints},
  year    = {2025},
  number  = {1},
  pages   = {29--61},
  volume  = {29},
  doi     = {10.7155/jgaa.v29i1.2975},
}

@Article{GBNH.ZME.2021,
  author  = {Sven Gedicke and Annika Bonerath and Benjamin Niedermann and Jan{-}Henrik Haunert},
  journal = {{IEEE} Transactions on Visualization and Computer Graphics},
  title   = {{Z}oomless {M}aps: {E}xternal {L}abeling {M}ethods for the {I}nteractive {E}xploration of {D}ense {P}oint {S}ets at a {F}ixed {M}ap {S}cale},
  year    = {2021},
  number  = {2},
  pages   = {1247--1256},
  volume  = {27},
  doi     = {10.1109/TVCG.2020.3030399},
}

@Article{FKP.Ctv.2010,
  author  = {Fernau, Henning and Kaufmann, Michael and Poths, Mathias},
  journal = {Journal of Computer and System Sciences},
  title   = {Comparing trees via crossing minimization},
  year    = {2010},
  number  = {7},
  pages   = {593--608},
  volume  = {76},
  doi     = {10.1016/j.jcss.2009.10.014},
}

@Article{DLM.SGD.2019,
  author  = {Didimo, Walter and Liotta, Giuseppe and Montecchiani, Fabrizio},
  journal = {ACM Computing Surveys},
  title   = {A Survey on Graph Drawing Beyond Planarity},
  year    = {2019},
  number  = {1},
  pages   = {1--37},
  volume  = {52},
  doi     = {10.1145/3301281},
}

@Article{DM.RGT.2018,
  author    = {Stephane Durocher and Debajyoti Mondal},
  journal   = {{SIAM} J. Discret. Math.},
  title     = {Relating Graph Thickness to Planar Layers and Bend Complexity},
  year      = {2018},
  number    = {4},
  pages     = {2703--2719},
  volume    = {32},
  _url      = {https://doi.org/10.1137/16M1110042},
  bibsource = {dblp computer science bibliography, https://dblp.org},
  biburl    = {https://dblp.org/rec/journals/siamdm/DurocherM18.bib},
  doi       = {10.1137/16M1110042},
  timestamp = {Sat, 25 Apr 2020 13:56:10 +0200},
}

@Article{MOS.TGS.1998,
  author    = {Petra Mutzel and Thomas Odenthal and Mark Scharbrodt},
  journal   = {Graphs Comb.},
  title     = {The Thickness of Graphs: {A} Survey},
  year      = {1998},
  number    = {1},
  pages     = {59--73},
  volume    = {14},
  _url      = {https://doi.org/10.1007/PL00007219},
  bibsource = {dblp computer science bibliography, https://dblp.org},
  biburl    = {https://dblp.org/rec/journals/gc/MutzelOS98.bib},
  doi       = {10.1007/PL00007219},
  timestamp = {Tue, 29 Dec 2020 18:23:55 +0100},
}

@InProceedings{JRRS.GT2.2023,
  author    = {Rahul Jain and Marco Ricci and Jonathan Rollin and Andr{\'{e}} Schulz},
  booktitle = {Proc. 39th International Symposium on Computational Geometry (SoCG'23)},
  title     = {On the Geometric Thickness of 2-Degenerate Graphs},
  year      = {2023},
  editor    = {Erin W. Chambers and Joachim Gudmundsson},
  pages     = {44:1--44:15},
  publisher = {Schloss Dagstuhl - Leibniz-Zentrum f{\"{u}}r Informatik},
  series    = {LIPIcs},
  volume    = {258},
  doi       = {10.4230/LIPICS.SOCG.2023.44},
}

@Article{GHN.MBL.2015,
  author  = {Andreas Gemsa and Jan{-}Henrik Haunert and Martin N{\"{o}}llenburg},
  journal = {{ACM} Transactions on Spatial Algorithms and Systems (TSAS)},
  title   = {{M}ultirow {B}oundary-{L}abeling {A}lgorithms for {P}anorama {I}mages},
  year    = {2015},
  number  = {1},
  pages   = {1:1--1:30},
  volume  = {1},
  doi     = {10.1145/2794299},
}

@article{GT.LTA.1985,
  author       = {Harold N. Gabow and
                  Robert Endre Tarjan},
  title        = {A Linear-Time Algorithm for a Special Case of Disjoint Set Union},
  journal      = {Journal of Computer and System Sciences},
  volume       = {30},
  number       = {2},
  pages        = {209--221},
  year         = {1985},
  doi          = {10.1016/0022-0000(85)90014-5}
}

@Article{MDZ.PAG.2001,
  author  = {Méndez Díaz -, Isabel and Zabala, Paula},
  journal = {Electronic Notes in Discrete Mathematics},
  title   = {A Polyhedral Approach for Graph Coloring},
  year    = {2001},
  pages   = {178--181},
  volume  = {7},
  doi     = {10.1016/s1571-0653(04)00254-9},
}

@InProceedings{JM.NIL.2018,
  author    = {Jabrayilov, Adalat and Mutzel, Petra},
  booktitle = {Proc. 13th Latin American Symposium on Theoretical Informatics (LATIN'18)},
  title     = {New Integer Linear Programming Models for the Vertex Coloring Problem},
  year      = {2018},
  pages     = {640--652},
  publisher = {Springer International Publishing},
  doi       = {10.1007/978-3-319-77404-6_47},
}

@Article{CCC.arf.2008,
  author    = {Campêlo, Manoel and Campos, Victor A. and Corrêa, Ricardo C.},
  journal   = {Discrete Applied Mathematics},
  title     = {On the asymmetric representatives formulation for the vertex coloring problem},
  year      = {2008},
  issn      = {0166-218X},
  month     = Apr,
  number    = {7},
  pages     = {1097--1111},
  volume    = {156},
  doi       = {10.1016/j.dam.2007.05.058},
  publisher = {Elsevier BV},
}

@Misc{Gurobi,
  author = {{Gurobi Optimization, LLC}},
  note   = {Accessed on 2026-05-06},
  title  = {{Gurobi Optimizer Reference Manual}},
  year   = {2026},
  url    = {https://www.gurobi.com},
}

@Misc{GeoPhylo,
  author    = {Jonathan Klawitter and Felix Klesen and Joris Y. Scholl and Thomas C. van Dijk and Alexander Zaft},
  note   = {Accessed on 2026-05-10},
  title  = {Algorithms to Visualize Geophylogenies -- GitHub Repository},
  year   = {2026},
  url    = {https://github.com/joklawitter/geophylo},
}

@article{Tahami2021,
  title = {Molecular phylogeny of cave dwelling {E}remogryllodes crickets (Orthoptera, Myrmecophilidae) across {Z}agros Mountains and {S}outhern {I}ran},
  volume = {50},
  DOI = {10.5038/1827-806x.50.2.2360},
  number = {2},
  journal = {International Journal of Speleology},
  author = {Tahami,  Mohadeseh Sadat and Hojat-Ansari,  Mina and Namyatova,  Anna and Sadeghi,  Saber},
  year = {2021},
  pages = {213–221}
}

@book{Gol.AGT.1980,
  author = {Martin Charles Golumbic},
  title = {Algorithmic Graph Theory and Perfect Graphs},
  doi = {10.1016/c2013-0-10739-8},
  publisher = {Elsevier},
  year = {1980}
}

@Article{Mac.Uce.1999,
  author  = {Lindsay W. MacDonald},
  journal = {IEEE Computer Graphics and Applications},
  title   = {Using color effectively in computer graphics},
  year    = {1999},
  number  = {4},
  pages   = {20--35},
  volume  = {19},
  doi     = {10.1109/38.773961},
}

@InCollection{War.CFC.2021,
  author    = {Ware, Colin},
  booktitle = {Information Visualization},
  publisher = {Morgan Kaufmann},
  title     = {Chapter Four - {C}olor},
  year      = {2021},
  chapter   = {four},
  edition   = {fourth},
  pages     = {95--141},
  doi       = {10.1016/b978-0-12-812875-6.00004-9},
}

@Misc{SUPPLEMENTARY_MATERIAL,
  author = {Thomas Depian and Thomas C. van Dijk and Nöllenburg, Martin},
  note   = {Source Code},
  title  = {Paged Geophylogenies: {A} Coloring Approach to
External Labeling with Tree Constraints -- {S}upplementary Material},
  year   = {2026},
  doi    = {10.17605/OSF.IO/RBQD3},
}

@inproceedings{dffgn-ptgt-25,
	Author = {Depian, Thomas and Fink, Simon D. and Firbas, Alexander and Ganian, Robert and Nöllenburg, Martin},
	Title = {Pathways to Tractability for Geometric Thickness},
	Booktitle = {Proc. 50th Conference on Current Trends in Theory and Practice of Computer Science (SOFSEM'25)},
	Year = {2025},
    Editor = {Kr{\'{a}}lovic, Rastislav and Kurkov{\'{a}}, Vera},
    Volume = {15538},
    Pages = {209--224},
    Series = {LNCS},
    Publisher = {Springer},
    Doi = {10.1007/978-3-031-82670-2_16},
}

\newpage
\appendix
\onlyShort{

\crefalias{section}{appendix}

\section{Omitted Details from \Cref{sec:tanglegram}}
\label{app:tanglegram}

We now give the full details for the following result:
\thmMLDS*
\label{thm:MLDS*}

\mldsDetails

\thmMLDS*

\section{Omitted Details from \Cref{sec:tree-consistent}}
\label{app:tree-consistent}
\subsection{Omitted Details from \Cref{sec:planar}}
We now provide the full details for our \BigO{n^5}-time DP to test planarity of a geophylogeny.

\planarSetup

\planarEfficient

\subsection{Omitted Proofs}
\planarCorrectness

\theoremPlanar*
\label{thm:planar*}
\ptheoremPlanar

\lemmaConsistentPlanar*
\label{lem:tree-consistent-planar*}
\plemmaConsistentPlanar

\theoremTreeConsistent*
\label{thm:tree-consistent*}
\ptheoremTreeConsistent

\section{Omitted Proofs from \Cref{sec:extensions}}

\theoremPageFixed*
\label{thm:page-fixed*}
\ptheoremPageFixed

\theoremPageSliding*
\label{thm:page-sliding*}
\ptheoremPageSliding

\newpage

}

\section{Omitted Plots from \Cref{sec:experiments}}
\label[appendix]{app:experiments}

\begin{figure}[!ht]
	\centering
	\includegraphics[page=1,width=\linewidth]{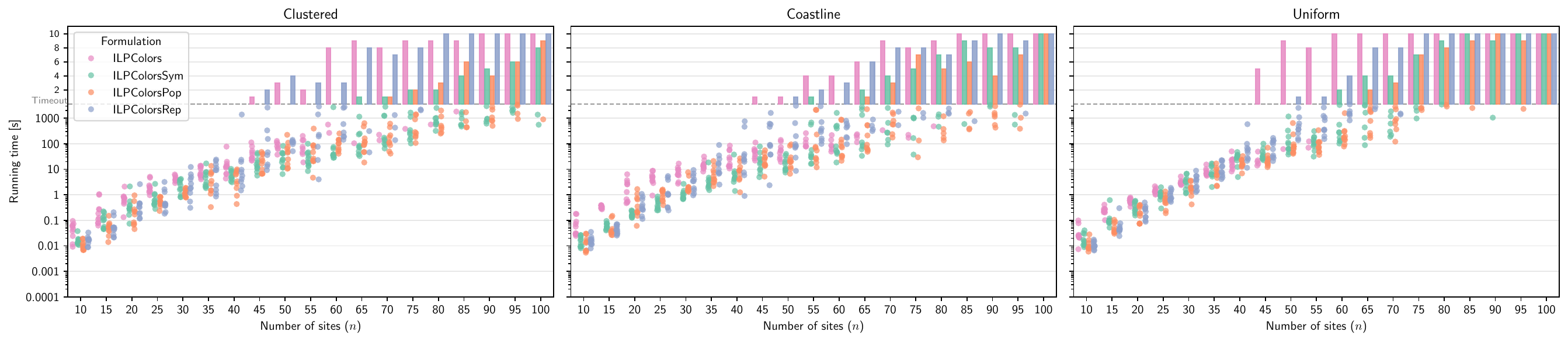}
	\caption{Same plot as \Cref{fig:formulations-running-times} but including uniform instances.}
	\label{fig:formulations-running-times-all}
\end{figure}

\begin{figure}[!ht]
	\centering
	\includegraphics[width=\linewidth]{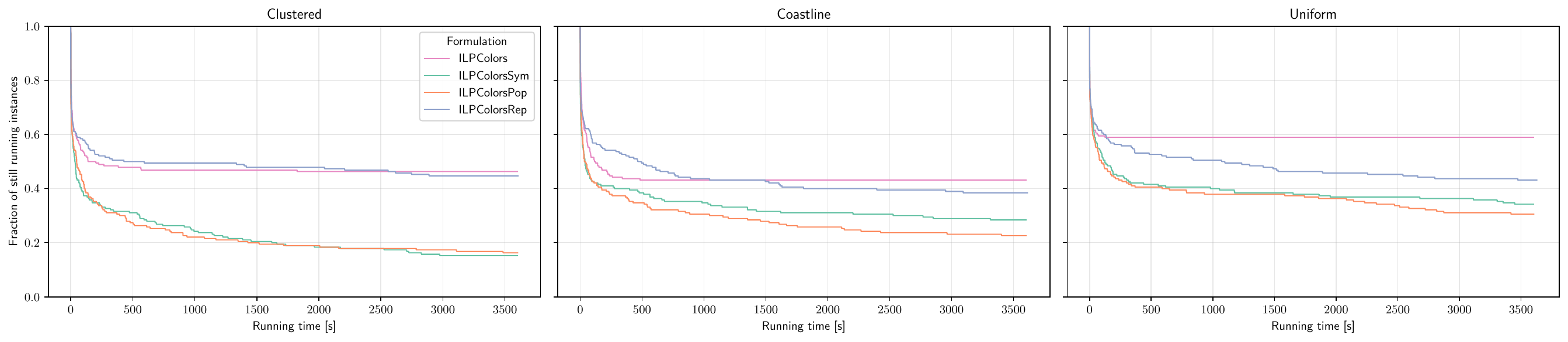}
	\caption{Same plot as \Cref{fig:kaplanmeier} but including uniform instances.}
    \label{fig:kaplanmeier-all}
\end{figure}

\begin{figure}[!ht]
	\centering
	\includegraphics[page=1,width=\linewidth]{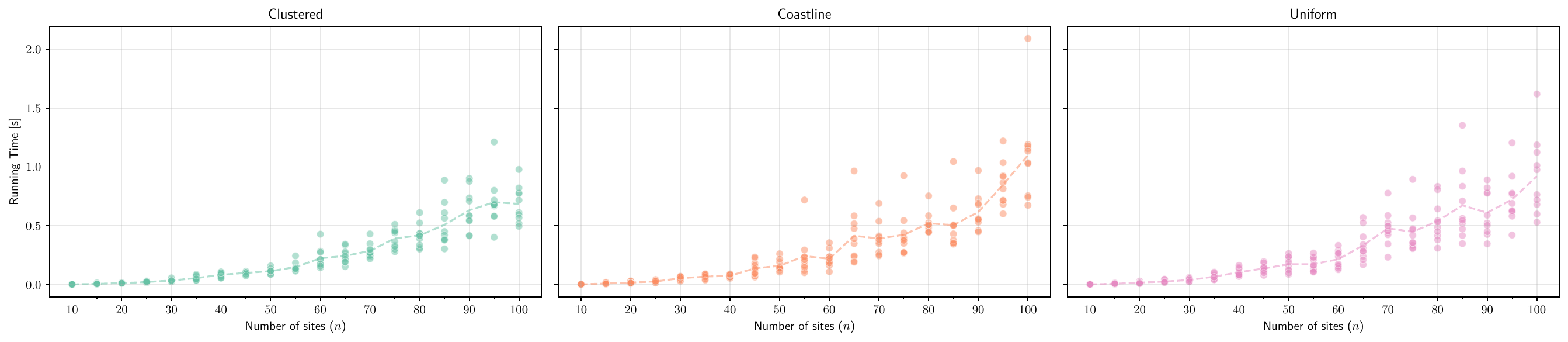}
	\caption{Same plot as \Cref{fig:treedp-running-times} but including uniform instances.}
	\label{fig:treedp-running-times-all}
\end{figure}

\begin{figure}[!ht]
	\centering
	\includegraphics[width=\linewidth]{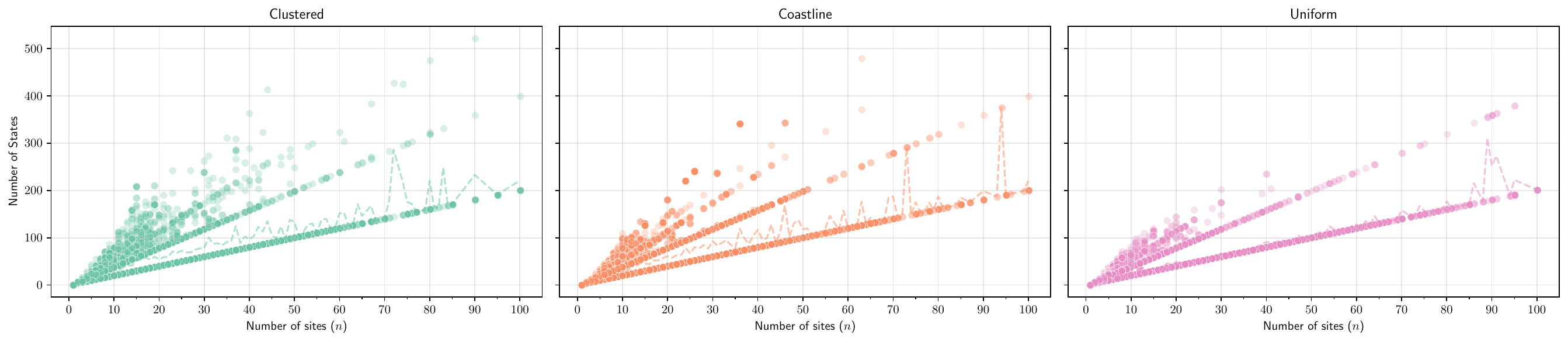}
	\caption{Same plot as \Cref{fig:treedp-states} but including uniform instances.}
    \label{fig:treedp-states-all}
\end{figure}

\begin{figure}[!ht]
	\centering
	\includegraphics[page=1,width=\linewidth]{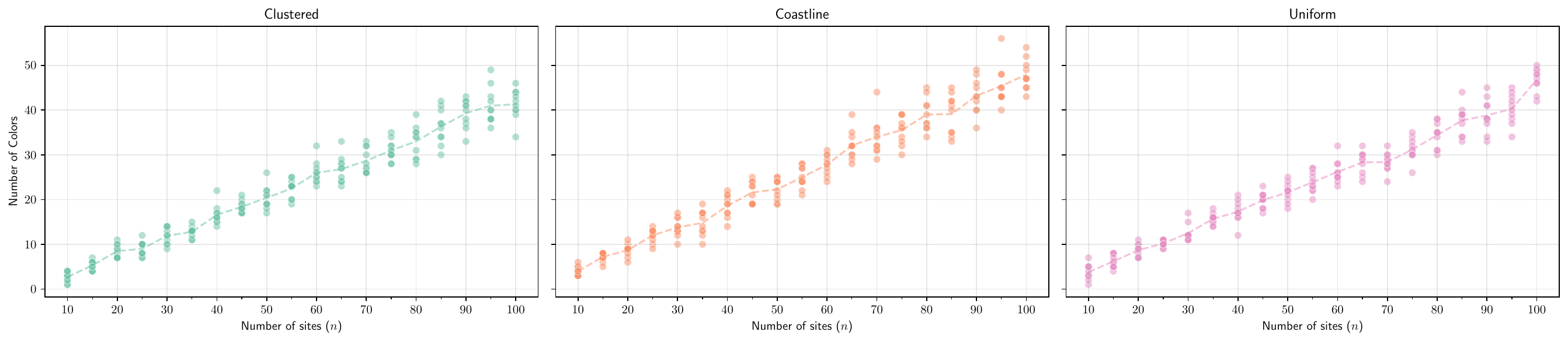}
	\caption{Same plot as \Cref{fig:tree-dp-colors} but including uniform instances.}
	\label{fig:tree-dp-colors-all}
\end{figure}

\begin{figure}[!ht]
	\centering
	\includegraphics[page=1,width=\linewidth]{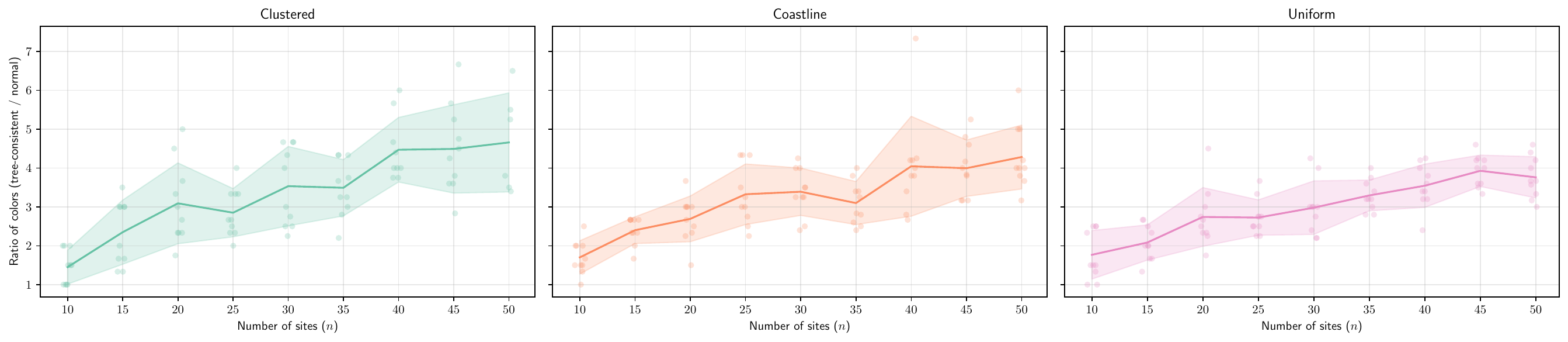}
	\caption{Same plot as \Cref{fig:tree-vs-normal-colors} but including uniform instances.}
	\label{fig:tree-vs-normal-colors-all}
\end{figure}

\begin{figure}[!ht]
	\centering
	\includegraphics[page=1,width=\linewidth]{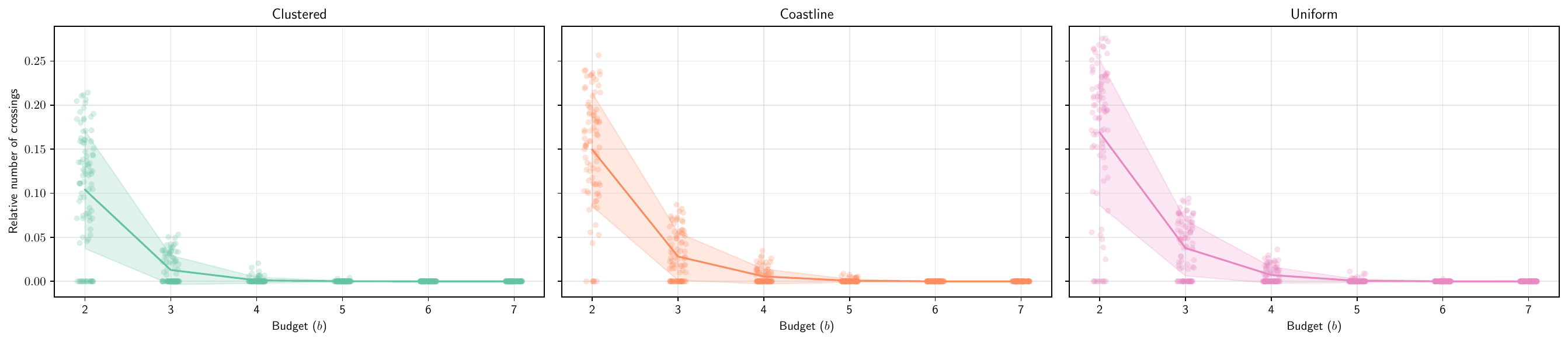}
	\caption{Same plot as \Cref{fig:tradeoff-all-colors} but including uniform instances.}
	\label{fig:tradeoff-all-colors-all}
\end{figure}

\begin{figure}[!ht]
	\centering
	\includegraphics[page=1,width=\linewidth]{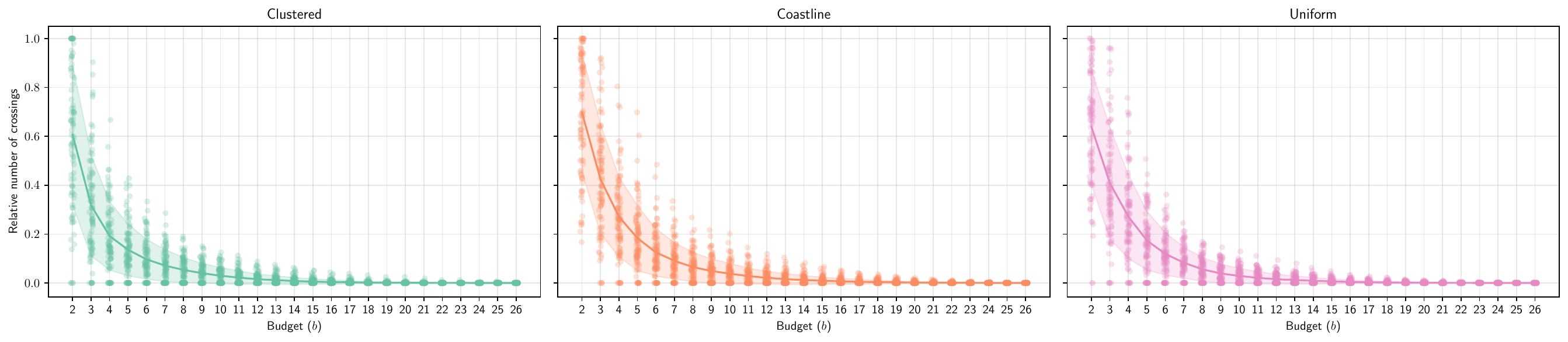}
	\caption{Same plot as \Cref{fig:tradeoff-tree-colors} but including uniform instances.}
	\label{fig:tradeoff-tree-colors-all}
\end{figure}

\end{document}